\documentclass[11 pt]{article}

\usepackage{amsmath,amssymb,amsthm,amscd,amsfonts}
\usepackage[numbers,square]{natbib}
\RequirePackage{hyperref}
\usepackage{graphicx}
\usepackage[text={17cm,24cm}]{geometry}

\DeclareMathOperator{\tr}{tr}

\def\rank{\mathop{\mathrm{rank}}}
\newtheorem{corollary}{Corollary}
\newtheorem{theorem}{Theorem}
\newtheorem{proposition}{Proposition}
\newtheorem{algorithm}{Algorithm}
\newtheorem{lemma}{Lemma}
\newtheorem{remark}{Remark}

\newcommand{\bt}{\begin{theorem}}
\newcommand{\et}{\end{theorem}}
\newcommand{\bl}{\begin{lemma}}
\newcommand{\el}{\end{lemma}}
\newcommand{\bp}{\begin{proposition}}
\newcommand{\ep}{\end{proposition}}
\newcommand{\bc}{\begin{corollary}}
\newcommand{\ec}{\end{corollary}}

\newcommand{\bd}{\begin{definition}\rm}
\newcommand{\ed}{\end{definition}}
\newcommand{\bex}{\begin{example}\rm}
\newcommand{\eex}{\end{example}}
\newcommand{\br}{\begin{remark}\rm}
\newcommand{\er}{\end{remark}}

\newcommand{\btbh}{\begin{table}[!ht]}
\newcommand{\etb}{\end{table}}
\newcommand{\bfgh}{\begin{figure}[!ht]}
\newcommand{\efg}{\end{figure}}

\newcommand{\bea}{\begin{eqnarray*}}
\newcommand{\eea}{\end{eqnarray*}}
\newcommand{\be}{\begin{eqnarray}}
\newcommand{\ee}{\end{eqnarray}}
\newcommand{\ve}{\varepsilon}

\def\what{\widehat}

\newcommand{\ra}{\rightarrow}

\def\spaceR{\mathsf{R}}

\def\rank{\mathop{\mathrm{rank}}}

\newcommand{\diag}{\mathop{\mathrm{diag}}}

\makeatletter
\def\adots{\mathinner{\mkern2mu\raise\p@\hbox{.}
\mkern2mu\raise4\p@\hbox{.}\mkern1mu
\raise7\p@\vbox{\kern7\p@\hbox{.}}\mkern1mu}}
\newcommand{\l@abcd}[2]{\hbox to\textwidth{#1\dotfill #2}}
\makeatother

\usepackage{euscript}

\newcommand{\rmF}{\mathrm{F}}

\newcommand{\rmT}{\mathrm{T}}

\newcommand{\tsI}{\mathbb{I}}

\newcommand{\tsL}{\mathbb{L}}

\newcommand{\tsN}{\mathbb{N}}

\newcommand{\tsR}{\mathbb{R}}
\newcommand{\tsS}{\mathbb{S}}

\newcommand{\tsX}{\mathbb{X}}
\newcommand{\tsY}{\mathbb{Y}}
\newcommand{\tsZ}{\mathbb{Z}}

\newcommand{\bfB}{\mathbf{B}}
\newcommand{\bfC}{\mathbf{C}}

\newcommand{\bfM}{\mathbf{M}}

\newcommand{\bfO}{\mathbf{O}}

\newcommand{\bfQ}{\mathbf{Q}}

\newcommand{\bfS}{\mathbf{S}}

\newcommand{\bfU}{\mathbf{U}}
\newcommand{\bfV}{\mathbf{V}}

\newcommand{\bfX}{\mathbf{X}}
\newcommand{\bfY}{\mathbf{Y}}
\newcommand{\bfZ}{\mathbf{Z}}

\newcommand{\calA}{\mathcal{A}}

\newcommand{\calH}{\mathcal{H}}

\newcommand{\calM}{\mathcal{M}}

\newcommand{\calT}{\mathcal{T}}

\newcommand{\sfR}{\mathsf{R}}

\newcommand{\sfT}{\mathsf{T}}

\newcommand{\sfX}{\mathsf{X}}

\newcommand{\sfZ}{\mathsf{Z}}
\newcommand{\bfx}{\textbf{x}}
\newcommand{\bfy}{\textbf{y}}
\newcommand{\bfz}{\textbf{z}}

\begin{document}

\title{Iterative algorithms for weighted and unweighted finite-rank time-series approximations}

\author{{Nikita} {Zvonarev}\footnote{St.Petersburg State University}, Nina Golyandina\footnote{St.Petersburg State University}}

\date{27.04.2015}
\maketitle

\begin{abstract}
The problem of time series approximation by series of finite rank is considered from the viewpoint of signal extraction. For signal estimation,
a weighted least-squares method is applied to the trajectory matrix of the considered time series. Matrix weights are chosen to obtain equal or approximately equal weights in the equivalent problem of time-series least-squares approximation.
Several new methods are suggested and examined together with the Cadzow's iterative method. The questions of convergence, computational complexity, and accuracy are considered for the proposed methods. The methods are compared on numeric examples.
\end{abstract}


\section{Introduction}
Consider the problem of extracting a signal $\tsS~=~(s_1, \ldots, s_N)$ from an observed noisy series $\tsX = \tsS + \tsN$, where $\tsS$ is governed by a \emph{linear recurrence relation} (LRR) of order $r$:
\begin{equation*}
s_n = \sum_{i = 1}^{r} a_i s_{n-i}, \quad n = r + 1, \ldots, N;\  a_r\neq 0.
\end{equation*}
Generally, series, which are governed by LRRs, may be written in a parametric form
\begin{equation} \label{parametricform}
s_n = \sum_i P_i(n) \exp(\alpha_i n) \cos(2 \pi \omega_i n + \psi_i),
\end{equation}
where $P_i(n)$ are polynomials of $n$. However, a parametric regression approach for the problem does not lead to accurate estimation of parameters due instability of estimates.

It is known that methods based on signal subspace estimation (subspace-based methods) work well \cite{Broomhead.King1986, Vautard.etal1992, Elsner.Tsonis1996, Golyandina.etal2001}. These subspace-based methods use the following approach. Let us fix a window length $L$, $1 < L < N$, set $K = N - L + 1$, and build the \emph{trajectory matrix} for the series $\tsS$:
\begin{equation*}
\bfS = \begin{pmatrix}
s_1 & s_2 & \ldots & s_K \\
s_2 & s_3 & \ldots & s_{K + 1} \\
\vdots & \vdots & \vdots & \vdots \\
s_L & s_{L + 1} & \ldots & s_N
\end{pmatrix}.
\end{equation*}
Note that $\bfS\in \calH$, where $\calH$ is the set of Hankel matrices with equal values on their anti-diagonals $i+j=\mathrm{const}$.
 Let $\tsS$ be governed by an LRR of order $r$, $r < \min(L, K)$, and be not governed by an LRR of smaller order. Then $\rank \bfS = r$ and therefore $\bfS$ is a Hankel matrix of low-rank $r$. The column space of $\bfS$, that is, the  signal subspace, provides estimates of $\alpha_i$ and $\omega_i$ in \eqref{parametricform} by the ESPRIT method \cite{Roy.Kailath1989, Golyandina.Zhigljavsky2012} applied to $\bfS$.

Let $\bfX$ be the trajectory matrix of the series $\tsX$. Then the problem of estimation of $\tsS$ and the signal subspace can be considered as a problem of approximation of the matrix $\bfX$ by a Hankel matrix of rank not larger than $r$:
\begin{equation}\label{introd_task}
\|\bfX - \bfY\|^2_\rmF \to \min_{\substack{\rank \bfY \le r \\ \bfY \in \calH}},
\end{equation}
where $\|\cdot\|_\rmF$ is the Frobenius norm.

Many papers are devoted to this problem, e.g., \cite{Cadzow1988, Markovsky2011, Usevich.Markovsky2014, Gillard.Zhigljavsky2013} among others, where the problem is called Structured Low-Rank Approximation. Numerical solutions of the problem are iterative; e.g., the Cadzow iterative method \cite{Cadzow1988} consists of alternating projections to the sets of Hankel matrices and of matrices of rank not larger than $r$. The target function is not unimodal in such class of problems, and convergence to the global minimum is not guaranteed; despite this, the problem \eqref{introd_task} is considered to be well-researched, though it still has many open questions.

Note that the problem \eqref{introd_task} is equivalent to the problem of weighted approximation of the series $\tsX = (x_1, \ldots, x_N)$:
\begin{equation}\label{introd_task_2}
\sum_{i = 1}^N w_i(x_i - y_i)^2 \to \min_{\substack{\tsY: \rank \bfY \le r \\ \bfY \in \calH}},
\end{equation}
where
\begin{equation}
\label{eq:w}
w_i = \begin{cases}
i & \text{for $i = 1, \ldots, L-1,$}\\
L & \text{for $i = L, \ldots, K,$}\\
N - i + 1 & \text{for $i = K + 1, \ldots, N$}
\end{cases},
\end{equation}
and $\bfY$ is the  trajectory matrix of the series $\tsY$.

The weights \eqref{eq:w} at both ends of the series are smaller than that in the center, i.e. the ordinary least-square problem \eqref{introd_task} for matrices corresponds to a weighted least-squares problem for series.

The aim of this paper is to consider methods which solve the problem \eqref{introd_task_2} with equal weights instead of $w_i$ and then to compare the constructed methods in terms of accuracy of the signal estimation. All described methods are iterative.
If one is interested in a signal estimate, which is not necessarily governed by an LRR, then the first iteration can be taken as a low-cost estimate of the signal.
Hence, the described methods are compared by accuracy of the signal estimation at the first iteration and in the limit. Note that Singular Spectrum Analysis (SSA) \cite{Broomhead.King1986, Vautard.etal1992, Elsner.Tsonis1996, Golyandina.etal2001, Ghil.etal2002, Golyandina.Zhigljavsky2012} applied to the problem of signal estimation can be represented as the first iteration of the Cadzow method.

The structure of the paper is as follows.  In Section~\ref{sec:lowrank_appr}, the problem of approximating a matrix by a Hankel rank-deficient matrix is considered. The common structure of iterative alternating-projection algorithms is described, approaches to construction of the projectors are given, the convergence theorem is proved.

In Section~\ref{sec:ts_matrices}, the relation between the problems of approximation of time series and of their trajectory matrices is described. The relationship between weights in equivalent weighted least-squares problems is also given. Section~\ref{sec:alg} contains the suggested time-series approximation algorithms. In Section~\ref{sec:simul}, a numeric comparison of algorithms on a typical simulated example is performed. Section~\ref{sec:ex_real} contains an example with analysis of real-life data.

The paper is summarized and conclusions are drawn in Section~\ref{sec:concl}. Supplementary results on SSA separability, which has a connection with the convergence rate, are proved in Appendix~\ref{sec:app}.

\section{Approximation by rank-deficient Hankel matrices}
\label{sec:lowrank_appr}
\subsection{Common scheme of iterations}
Consider the problem of projecting a point $\bfx$ to a set $\calH~ \cap~\calM$ in a Hilbert space $\sfX$ with a inner product $\langle \cdot, \cdot \rangle$, where $\calH$ and $\calM$ are closed under the limit operation, $\calH$ is linear subspace, while $\calM$ is closed with respect to scalar multiplication, i.e.
if $\bfz \in \calM$, then $\alpha \bfz\in \calM$ for any $\alpha$. Note that  $\calM$ is not necessarily a linear space or a convex set.

Thus, the problem is formulated as
\be
\label{eq:gen_task}
\|\bfx - \bfy\| \to \min_\bfy \mbox{\ over\ } \bfy \in \calH \cap \calM,
\ee
where $\|\cdot\|$ is the norm corresponding to the inner product.

To present the algorithm's scheme for the solution of this problem, let us introduce the projectors to the subset $\calM$ and subspace $\calH$ with respect to the norm $\|\cdot\|$: $\Pi_{\calM}$ is the projector to $\calM$,
$\Pi_{\calH}$ is the projector to $\calH$.
Note that if the projection to $\calM$ is not uniquely defined, then we suppose that in the case of ambiguity any closest point is chosen.
The projector to $\calH$ is evidently orthogonal, while $\Pi_{\calM}$ is orthogonal due to the following proposition.

\begin{proposition} \label{prop:pythaprop}
	Let $\sfX$ be a Hilbert space, $\calM \subset \sfX$ be a subset closed with respect to scalar multiplication, $\Pi_\calM$ be the projection operator to $\calM$. Then for any $\bfx \in \sfX$ the following equation (``Pythagorean equality'') is true: $\|\bfx\|^2~=~\|\bfx~-~\Pi_\calM \bfx\|^2~+~\|\Pi_\calM \bfx\|^2$.
\end{proposition}

\begin{proof}
	Define $\bfy = \Pi_\calM \bfx$. Since
	\begin{gather*}
	\|\bfx\|^2 = \|\bfx - \bfy \|^2 + \|\bfy \|^2 + 2 \langle \bfx - \bfy, \bfy \rangle,
	\end{gather*}
    we should prove that $\langle \bfx - \bfy, \bfy \rangle = 0$.
	Assume the opposite: $\langle \bfx - \bfy, \bfy \rangle \ne 0$. Then for
	\begin{equation*}
	\gamma = \frac{\langle \bfx, \bfy \rangle}{\langle \bfy, \bfy \rangle}
	\end{equation*}
    $\langle \bfx - \gamma \bfy, \gamma \bfy \rangle = 0$
    and  therefore $\|\bfx - \bfy\|^2 > \|\bfx - \gamma \bfy\|^2$:
	\begin{gather*}
	\|\bfx - \bfy\|^2 - \|\bfx - \gamma \bfy\|^2 = \\\langle \bfy, \bfy \rangle - 2 \langle \bfx, \bfy \rangle + \frac{\langle \bfx, \bfy \rangle ^ 2}{\langle \bfy, \bfy \rangle} =
    \frac{\langle \bfx - \bfy, \bfy \rangle^2}{\langle \bfy, \bfy \rangle} > 0.
	\end{gather*}
	Since 	$\gamma \bfy$ lies in $\calM$ according to the property of $\calM$,
    the contradiction with the fact that $\bfy = \Pi_\calM \bfx$ is the closest point to $\bfx$ is acquired.
\end{proof}

\begin{remark}
\label{rem:adj}
The proof of Proposition~\ref{prop:pythaprop} yields that for any $\bfy\in \calM$ one can perform an adjustment $\calA(\bfy)=\frac{\langle \bfx, \bfy \rangle}{\langle \bfy, \bfy \rangle} \bfy \in \calM$ such that $\calA(\bfy)$
is not further from $\bfx$ than the original $\bfy$. Moreover, $\calA(\bfy)$ is
orthogonal to $\bfx - \calA(\bfy)$.
\end{remark}

Let us consider the \emph{iterative method of alternating projections} for the problem \eqref{eq:gen_task},
which is given by the following iteration step:
\be
\label{eq:iter}
\bfy_{k+1}=\Pi_\calH \Pi_{\calM} \bfy_{k}, \mbox{\ where\ } \bfy_{0}=\bfx.
\ee

In the following theorem, we investigate convergence of the sequence~\eqref{eq:iter}.

\begin{theorem}
	\label{th:converg}
		Let the conditions of Proposition~\ref{prop:pythaprop} be fulfilled and also the set $\calM$ and the space $\calH$ be closed under the limit operation. Then
	\begin{enumerate}
		\item $\|\bfy_k - \Pi_{\calM}\bfy_k\| \to 0$ as $k \to +\infty$, $\|\Pi_{\calM}\bfy_k - \bfy_{k+1}\| \to 0$ as $k \to +\infty$.
		\item Let $\calM \cap B_1$ be a compact set, where $B_1=\{\bfz: \|\bfz\|~\le~1\}$ is the closed unit ball. Then there exists a convergent subsequence of points $\bfy_{i_1}, \bfy_{i_2}, \ldots$ such that its limit $\bfy^*$ belongs to $\calM \cap \calH$.
	\end{enumerate}
\end{theorem}
\begin{proof}
	Let us use the following inequalities:
	\begin{multline}
	\label{eq:chuprop}
	\|\bfy_k - \Pi_{\calM} \bfy_k\| \ge \|\Pi_{\calM} \bfy_k - \bfy_{k + 1}\| \ge \\ \|\bfy_{k+1} - \Pi_{\calM} \bfy_{k + 1}\|.
	\end{multline}
Indeed, since the projection $\Pi_\calM \bfz$ is not further from $\bfz$ than any other point from $\calM$ and the similar statement is valid for
$\Pi_\calH$, we have $\|\Pi_{\calM} \bfy_k - \bfz\| \ge \|\bfz - \Pi_{\calM} \bfz\|$, where $\bfz=\bfy_{k+1}$, and $\|\bfy_k - \bfz\| \ge \|\bfz - \Pi_{\calH} \bfz\|$, where $\bfz=\Pi_{\calM} \bfy_k$.
	\begin{enumerate}
		\item According to inequalities \eqref{eq:chuprop}, the sequences $\|\bfy_k~-~\Pi_{\calM} \bfy_k\|$, $k = 1, 2, \ldots$, and $\|\Pi_{\calM} \bfy_k - \bfy_{k + 1}\|$, $k = 1, 2, \ldots$, are non-increasing. It is obvious that they are limited below by zero. Therefore, they have the same limit $c$ due to \eqref{eq:chuprop}.
		
		Let us prove that $c = 0$ assuming the opposite $c > 0$. Then there exists $d > 0$ such that $\|\bfy_k - \Pi_{\calM} \bfy_k\| > d$ and $\|\Pi_{\calM} \bfy_k - \bfy_{k + 1}\| > d$ for any $k = 1, 2, \ldots$. In accordance to Proposition~\ref{prop:pythaprop}, the following equality is valid: $\|\bfy_k \|^2~=~\|\bfy_k~-~\Pi_{\calM} \bfy_k\|^2~+~\|\Pi_{\calM} \bfy_k \|^2$. Since the space $\calH$ is linear, the following equality is valid too:
		$\|\Pi_{\calM} \bfy_k \|^2~=\|\Pi_{\calM} \bfy_k~-~\Pi_\calH \Pi_{\calM} \bfy_k\|^2~+~\|\Pi_\calH \Pi_{\calM} \bfy_k \|^2 = \|\Pi_{\calM} \bfy_k~-~\bfy_{k+1}\|^2~+~\|\bfy_{k+1} \|^2$. Therefore,
		\begin{multline*}
		\|\bfy_k\|^2 = \|\Pi_{\calM} \bfy_k\|^2 + \|\bfy_k - \Pi_{\calM} \bfy_k\|^2 =\\ \|\bfy_k - \Pi_{\calM} \bfy_k\|^2 + \|\Pi_{\calM} \bfy_k - \bfy_{k + 1}\|^2 + \|\bfy_{k + 1}\|^2.
		\end{multline*}
		Thus, $\|\bfy_{k+1}\|^2 < \|\bfy_k\|^2 - 2d^2$. Expanding this inequality by the same way, we obtain that $\|\bfy_{k+j}\|^2 < \|\bfy_k\|^2 - 2 j d^2$ for any $j = 1, 2, \ldots$. Choose $k = 1$, and $j = \lceil \|\bfy_k\|^2 / (2d^2) \rceil + 1$. Then $\|\bfy_{k+j}\|^2 < 0$, which is impossible. Thus, $c=0$.
		\item Consider the sequence $\Pi_{\calM} \bfy_k$, $k = 1, 2, \ldots$, which is bounded, since $\|\Pi_{\calM} \bfz\| \le \|\bfz\|$ (by Proposition \ref{prop:pythaprop}) and $\|\Pi_{\calH} \bfz\| \le \|\bfz\|$ for any $\bfz \in \sfX$. The sequence belongs to a compact set, since $\calM$ is closed with respect to scalar multiplication, and we can resize the unit ball to cover the sequence. Then a convergent subsequence $(\Pi_{\calM} \bfy_{i_k})$ can be chosen; denote by $\bfy^*\in\calM$ its limit and notice that $\|\Pi_{\calM} \bfy_{i_k} - \bfy_{i_k + 1}\| = \|\Pi_{\calM} \bfy_{i_k} - \Pi_\calH \Pi_{\calM} \bfy_{i_k}\| \to 0$ as $k \to + \infty$. Since  $\calH$ is closed, and $\sfX$ is a Banach space, the projector $\Pi_\calH$ is a continuous mapping. Taking into consideration that $\|\bfz - \Pi_\calH \bfz\|$ is a composition of continuous mappings, we obtain that $\|\bfy^* - \Pi_\calH \bfy^*\| = 0$, $\bfy^* \in \calM \cap \calH$. Finally, $\Pi_\calH$ is a continuous mapping and therefore the sequence $(\Pi_\calH \Pi_{\calM} \bfy_{i_k})$ converges to $\bfy^*$. Thus, $\bfy_{i_k + 1}$ is the required subsequence.
	\end{enumerate}
\end{proof}

Actually, Proposition \ref{prop:pythaprop} was in fact proved in \cite{Gillard.Zhigljavsky2013} for a particular case, while inequalities \eqref{eq:chuprop} are extensions of \cite[inequalities (4.1)]{Chu.etal2003}.

\medskip
Let us apply Theorem~\ref{th:converg} to the case of matrix approximation by rank-deficient Hankel matrices. Let $\sfX = \spaceR^{L\times K}$, i.e. $\sfX$ be the space of matrices of size $L \times K$ equipped with some inner product, $\calH \subset \spaceR^{L\times K}$ be the space of Hankel matrices, $\calM = \calM_r\subset \spaceR^{L\times K}$ be the set of matrices of rank not larger than $r$. Then the iterative step \ref{eq:iter} for method of alternating projections has the following form:
\begin{equation*}
\bfY_{k+1}=\Pi_\calH \Pi_{\calM_r} \bfY_{k}, \mbox{\ where\ } \bfY_{0}=\bfX \in \spaceR^{L\times K}.
\end{equation*}

It is well known that the set $\calM_r$ is closed with respect to the conventional Frobenius norm and therefore is closed
to any norm, since in the matrix space all the norm are equivalent.
The closed unit ball is obviously a compact set in finite-dimensional Euclidean space.
Therefore the conclusion of Theorem~\ref{th:converg} holds.
Note that the existence of a convergent subsequence can be deduced from \cite{Cadzow1988}.
However, our proof of this fact is based on different assumptions; in particular, we stress on the Pythagorean equality for projections to sets which are closed with respect to multiplication.

In this paper, we consider norms (semi-norms) in $\sfX$ generated by weighted Frobenius inner products in the form, which is parameterized by a matrix $\bfM$ with positive (non-negative) entries $m_{i,j}$:
\be
\label{eq:w_inner_prod}
\langle\bfY, \bfZ\rangle_\bfM = \sum_{l = 1}^L \sum_{k = 1}^K m_{l, k} y_{l, k} z_{l, k}.
\ee
Therefore, the conclusion of Theorem~\ref{th:converg} holds if the weights $m_{i,j}$ are positive.

\subsection{Evaluation of projections}
Let us consider the weighted norm $\|\cdot\|_\bfM$ generated by \eqref{eq:w_inner_prod}, that is, $\|\bfX\|^2 = \|\bfX\|^2_\bfM = \sum_{l = 1}^L \sum_{k = 1}^K m_{l, k} x^2_{l, k}$.

\subsubsection{Projector $\Pi_\calH$.}
\label{sec:projH}
It is easy to show that $\Pi_\calH$
can be evaluated explicitly using the following proposition.

\begin{proposition}
	For $\widehat{\bfY}=\Pi_\calH \bfY$ we have
	\begin{equation*}
	\hat{y}_{ij} = \frac{\sum_{l,k:\, l+k=i+j} m_{l,k} y_{l,k}}{\sum_{l,k:\, l+k=i+j} m_{l,k}}.
	\end{equation*}
\end{proposition}

It is impossible to derive an explicit form of $\Pi_{\calM_r}$ in the case of arbitrary weights.
Consider one specific case and suggest an iterative approach to the general case.

\subsubsection{Case of the explicit form of the projector $\Pi_{\calM_r}$.}
\label{sec:projMr}
\label{sec:obliqueSVD}
For equal weights $m_{ij}=1$, denote $\Pi_r=\Pi_{\calM_r}$.
It is well-known that the projector $\Pi_{r} \bfY$ can be evaluated as the sum of $r$ leading components of the singular value decomposition (SVD) of the matrix $\bfY$. More precisely, let $L\le K$ for simplicity and $\bfY = \bfU \mathbf{\Sigma} \bfV^\rmT$ be the SVD, where $\bfU$ is an orthogonal matrix of size $L \times L$, $\mathbf{\Sigma}$ is a quasi-diagonal matrix of size $L \times K$ with non-negative diagonal elements $(\sigma_1, \ldots, \sigma_L)$ in non-increasing order, and $\bfV$ is an orthogonal matrix of size $K \times K$. Denote by $\mathbf{\Sigma}_r = (\sigma^r_{l k})$ the following matrix:
\begin{equation*}
\sigma^r_{i j} = \begin{cases}
\sigma_i & \text{if $i = j, i \le r,$}\\
0 & \text{otherwise}.
\end{cases}
\end{equation*}
Then the projection can be evaluated as $\Pi_{r} \bfY  = \bfU \mathbf{\Sigma}_r \bfV^\rmT$.
The next proposition describes the case when evaluation of a projector $\Pi_{\calM_r}$ is reduced to application of the projector $\Pi_r$.

\begin{proposition}
	\label{prop:projS}
	Let there exist a symmetric positive semidefinite matrix $\bfC$ of size $K \times K$ such that for a given $\bfM$ the  equality $\|\bfZ\|_\bfM^2 = \tr(\bfZ \bfC \bfZ^\rmT)$ holds for any matrix $\bfZ\in \spaceR^{L\times K}$.
	Suppose that the column space of a matrix $\bfY$ lies in the column space of the matrix $\bfC$.
	Then
	\be
	\label{eq:PiMr}
	\Pi_{\calM_r} \bfY = (\Pi_r \bfB) (\bfO_\bfC^{\rmT})^\dagger,
	\ee
	where $\bfO_\bfC$ is a matrix such that $\bfC = \bfO_\bfC^{\rmT}\bfO_\bfC$,
	$\bfB = \bfY \bfO_\bfC^{\rmT}$, $(\bfO_\bfC^{\rmT})^\dagger$ denotes  Moore-Penrose pseudoinverse to the matrix $\bfO_\bfC^{\rmT}$.
\end{proposition}
\begin{proof}
	The proof is a direct consequence of the fact that the considered norm is generated by an oblique inner product in the row space of $\bfY$, see details in \cite{Golyandina2013,Allen2014}.
\end{proof}

\begin{remark}
\label{rem:diagC}
	In fact, the condition $\|\bfZ\|_\bfM^2 = \tr(\bfZ \bfC \bfZ^\rmT)$ of Proposition~\ref{prop:projS} can be fulfilled only if $\bfC$ is diagonal and $\bfM$ has a specific form, see Proposition~\ref{prop:equiv_tasks}.
\end{remark}

\subsubsection{The projector $\Pi_{\calM_r}$ in the general case.}
Since the projector can not be found explicitly for arbitrary weights $m_{ij}$, iterative algorithms are used in the general case.
One of these algorithms is described in \cite{Srebro2003}. Denote by $\odot$ the element-wise matrix product.

\begin{algorithm}
	\label{alg:weightedSVD}
	\textbf{Input}: initial matrix $\bfY$, rank $r$, weight matrix $\bfM$,
	stop criterion STOP.
	
	\textbf{Result}:
	Matrix $\widehat\bfY$ as an estimate of $\Pi_{\calM_r} \bfY$.
	
	\begin{enumerate}
		\item
		$\bfY_0 = \bfY$, $k=0$.
		\item
		$\bfY_{k+1} = \Pi_r(\bfY \odot \bfM + \bfY_{k} \odot (\bfQ -  \bfM))$, where
		$\bfQ \in \sfR^{L \times K}$ is the matrix of all ones;
        \quad $k\leftarrow k+1$.
		\item
		If STOP, then $\widehat\bfY = \bfY_k$; else go to 2.
	\end{enumerate}
\end{algorithm}

Note that in the case, when $m_{ij}$ are equal to either 0 or 1, Algorithm~\ref{alg:weightedSVD} is an EM-algorithm \cite{Srebro2003};
hence, properties of EM-algorithms are carried out and the sequence $\bfY_k$ converges to a local minimum. Formally, it does not matter what values are in $\bfY$ at positions of zero weights. However, these values can influence the algorithm's convergence rate and the limiting values.

\section{Time series and problem of matrix approximation}
\label{sec:ts_matrices}
\subsection{Problem statement for time series}
\label{sec:ts}
Consider a time series $\tsX = (x_1, \ldots, x_N)$ of length $N \ge 3$. Let us fix a window length $L$, $1 < L < N$, denote $K = N - L + 1$. Also consider a sequence of \emph{$L$-lagged vectors}:
\begin{equation}\label{l_lagged}
X_i = (x_i, \ldots, x_{i + L - 1})^\rmT, \qquad i = 1, \ldots, K.
\end{equation}
Define an \emph{$L$-trajectory matrix} of the series $\tsX$ as $\bfX = [X_1:\ldots:X_K]$.

Suppose that $0 < r \le L$. We say that the series $\tsX$ \emph{has $L$-rank $r$} if its $L$-trajectory matrix $\bfX$ has rank $r$.

Note that the series $\tsX$ can have $L$-rank $r$ only when
\begin{equation}
r \le \min(L, K). \label{min_condition}
\end{equation}

Further we suppose that $L$ is not larger than $K$, since the problems of approximation of $\bfX$ and $\bfX^\rmT$ coincide.

Let $\sfX_N$ be the set of time series of length $N$, $\sfX_N^r$ be the set of time series of length $N$ which has $L$-rank not larger than $r$. For a given time series $\tsX \in \sfX_N$, a window length $L$, $1 < L < N$, and a rank $r$ satisfying condition \eqref{min_condition}, consider the problem:
\begin{equation} \label{L-rank_task}
f_q(\tsY) \to \min_{\tsY \in \sfX_N^r}, \quad f_q(\tsY) = \sum \limits_{i=1}^N q_i(x_i - y_i)^2,
\end{equation}
where $\tsY = (y_1, \ldots, y_N)$ and $q_1, \ldots, q_N$ are some non-negative weights,
$q_i \ge 0$, $i = 1, \ldots, N$. The squared Euclidean distance to $\tsX$ in $\sfR^N$ is one of reasonable target functions. It coincides with $f_q(\tsY)$ when $q_i = 1$, $i = 1, \ldots, N$.

\paragraph*{Adjustment.} Let an estimate $\tsY \in \sfX_N^r$ of the solution of the problem \eqref{L-rank_task} for approximation of $\tsX \in \sfX_N$ be obtained. Then, according to Remark~\ref{rem:adj}, the estimate can be adjusted to obtain a better estimate $\tsY^*=\calA(\tsY)$, which is called \emph{an adjustment of $\tsY$}.

\subsection{Equivalent target functions}

Let $\tsX = (x_1, \ldots, x_N) \in \sfX_N$ be a time series of length $N$, $\bfX = (\hat x_{l,k}) \in \calH$. Then there exists a one-to-one mapping $\calT$ between $\sfX_N$ and $\calH$, which can be written as
\begin{equation*}
\calT(\tsX) = \bfX,\ \text{where} \ \hat x_{l, k} = x_{l + k - 1}.
\end{equation*}

Due to this one-to-one mapping, the problem~\eqref{L-rank_task} of time series approximation can be expressed in terms of matrices.

In the space $\sfX_N$ of time series, the target function \eqref{L-rank_task} can be given explicitly $f_q(\tsY)=\|\tsY-\tsX\|_q^2$ using a (semi)inner product
\begin{equation}
\label{eq:norm_ser}
\langle\tsY,\tsZ\rangle_q = \sum_{i = 1}^N q_i y_i z_i,
\end{equation}
where $q_i$ are positive (non-negative) weights.

Consider two (semi)inner products in the space $\sfR^{L \times K}$ of matrices which are extensions of the conventional Frobenius inner product.

Denote, as before,
\begin{equation}
\label{eq:norm1M}
\langle\bfY,\bfZ\rangle_{1,\bfM} = \langle\bfY,\bfZ\rangle_{\bfM} = \sum_{l = 1}^L \sum_{k=1}^K m_{l,k} y_{l,k} z_{l,k}.
\end{equation}
for a matrix $\bfM\in \spaceR^{L\times K}$ with positive (non-negative) elements and also
\begin{equation}
\label{eq:norm2S}
\langle\bfY,\bfZ\rangle_{2,\bfC} = \tr(\bfY \bfC \bfZ^\rmT)
\end{equation}
for a positive (semi)definite symmetric matrix $\bfC \in \spaceR^{K\times K}$.

Note that if the matrix $\bfM$ consists of all ones, i.e. $m_{i,j}=1$,
and if $\bfC$ is the identity matrix, then both inner products coincide with the standard Frobenius inner product.

\begin{proposition}
	\label{prop:equiv_tasks}
	1. Let $\bfY = \calT(\tsY)$,  $\bfZ = \calT(\tsZ)$. Then $\langle\tsY,\tsZ\rangle_q= \langle \bfY,\bfZ \rangle_{1,\bfM}$ if and only if
	\begin{equation}\label{qi_mi}
	q_i = \sum_{\substack{1 \le l \le L \\ 1 \le k \le K \\ l+k-1=i}} m_{l,k}.
	\end{equation}
	
	2. The equality $\langle\bfY,\bfZ\rangle_{1,\bfM}= \langle\bfY,\bfZ\rangle_{2,\bfC}$ is valid if and only if
the matrix $\bfC=\diag(c_1,\ldots,c_K)$ and
	\begin{equation}\label{sk_mlk}
	m_{l,k}=c_k.
	\end{equation}
\end{proposition}
\begin{proof}
	To prove the first statement, note that
	\begin{equation*}
	\langle \bfY, \bfZ \rangle_{1,\bfM} = \sum_{i = 1}^L \sum_{j = 1}^K m_{i,j} y_{i + j - 1} z_{i + j - 1}.
	\end{equation*}
The proof of the second statement is a consequence of the fact that only for a diagonal matrix $\bfC$ the corresponding inner product has a form appropriate to \eqref{eq:norm1M} (see also Remark~\ref{rem:diagC}):
	\begin{equation*}
	\langle \bfY, \bfZ \rangle_{2,\bfC} = \sum_{l=1}^L \sum_{k=1}^K c_k y_{l,k} z_{l, k}.
	\end{equation*}
\end{proof}

\begin{corollary}
	\label{cor:base_weights}
	If $m_{i,j}=1$, $i =1, \ldots, L$, $j = 1, \ldots, K$, then the equivalent series weights $q_i$, $i = 1, \ldots, N$, given by \eqref{qi_mi} are equal to $w_i$ introduced in \eqref{eq:w}.
\end{corollary}

Note that the matrix norm $\|\cdot\|_{2, \bfC}$ with a diagonal matrix $\bfC$ is a particular case of the norm $\|\cdot\|_{1, \bfM}$.
However, this particular case is of special interest, since the corresponding approximation problem can be solved by means of the ordinary SVD, see Proposition~\ref{prop:projS}.

\begin{remark}
	\label{rem:2tasks}
	If the condition~\eqref{qi_mi} is carried out and all weights $q_i$ and $m_{i,j}$ are positive, then the problem~\eqref{L-rank_task}
	is equivalent to the problem
	\begin{multline}
	\label{rank_task}
	f_\bfM(\bfY) \to \min_{\bfY \in \calM_r \cap \calH}, \\ f^2_\bfM(\bfY) = \|\bfX-\bfY\|^2_{1,\bfM} = \sum_{l = 1}^L \sum_{k=1}^K m_{l,k} (x_{l,k} - y_{l,k})^2.
	\end{multline}
\end{remark}

\section{Algorithms}
\label{sec:alg}
In this section we suggest a range of algorithms for solving the problem~\eqref{L-rank_task}.
In the model of series $\tsX=\tsS+\tsN$, where $\tsS$ is a time series of finite rank $r$ and $\tsN$ is a noise series, results of the algorithms serve as  estimates of the signal $\tsS$.

\subsection{Cadzow iterations}
The aim of the Cadzow algorithm \cite{Cadzow1988} is the least-squares approximation \eqref{rank_task} of the trajectory matrix of a series with respect to the norm $\|\cdot\|_{1, \bfM}$ with the weights $m_{ij}=1$ (i.e. the algorithm solves the problem \eqref{introd_task}, which, by Corollary \ref{cor:base_weights}, corresponds to the problem~\eqref{introd_task_2} (or, the same, to the problem~\eqref{L-rank_task} with the weights $q_i=w_i$ given in \eqref{eq:w}). The drawback of this algorithm consists in the unequal series weights $w_i$: they are larger in the center than at both ends of the time series. Note that smaller window lengths leads to more uniform weights.

Note that in the case of unit weights $m_{ij}=1$, the projections $\Pi_\calH$ and $\Pi_{\calM_r}=\Pi_{r}$ can be easily calculated, see Sections~\ref{sec:projH}
and~\ref{sec:projMr}.

\begin{algorithm}[Cadzow iterations]
	\textbf{Input}: Time series $\tsX$, window length $L$, rank $r$,
	stop rule STOP1 (e.g., given by quantity of iterations).
	
	\textbf{Result}:
	Approximation $\widehat\tsS$ of time series $\tsX$ by finite-rank series of rank $r$.
	
	\begin{enumerate}
		\item
		$\bfY_0 = \calT \tsX$, $k=0$.
		\item
		$\bfY_{k+1} = \Pi_\calH  \Pi_{r} \bfY_{k}$, $k\leftarrow k+1$.
		\item
		If STOP1, then $\widehat\tsS = \calT^{-1} \bfY_k$; else go to 2.
	\end{enumerate}
\end{algorithm}

\subsection{Weighted Cadzow iterations}

Let $q_{i}=1$, $i = 1, \ldots, N$, be chosen in \eqref{L-rank_task}. According to Proposition ~\ref{prop:equiv_tasks}, the problem \eqref{L-rank_task} is equivalent to the problem \eqref{rank_task} with weights
\begin{equation}
\label{Mw}
m_{l, k} = \frac{1}{w_{l + k - 1}},
\end{equation}
where $w_i$ are introduced in \eqref{eq:w}.

\begin{algorithm}[Weighted Cadzow iterations]\label{alg:WCIt}
	\textbf{Input}: Time series $\tsX$, window length $L$, rank $r$,
	stop rules STOP1 for outer iterations and STOP2 for inner iterations.
	
	\textbf{Result}:
	Approximation $\widehat\tsS$ of time series $\tsX$ by finite-rank series of rank $r$.
	
	\begin{enumerate}
		\item
		$\bfY_0 = \calT \tsX$, $k=0$.
		\item
		Obtain $\widehat\bfZ$ using Algorithm ~\ref{alg:weightedSVD} applied to $\bfY_k$ for estimation of $\Pi_{\calM_r} \bfY_{k}$ with stop criterion STOP2.
		\item
		$\bfY_{k+1} = \Pi_\calH  \widehat\bfZ$, $k\leftarrow k+1$.
		\item
		If STOP1, then $\widehat\tsS = \calT^{-1} \bfY_k$; else go to 2.
	\end{enumerate}
\end{algorithm}

\subsection{Extended Cadzow iterations}
Let us introduce the Extended Cadzow algorithm, which presents a different approach to the problem \eqref{L-rank_task} with equal weights than the Weighted Cadzow algorithm does.
Formally, let the series $\tsX$ be extended to both sides on $L-1$ measurements with some values having zero weights, i.e., the added measurements are considered as gaps.
Thus, the length of the extended series $\widetilde\tsX$ is $N+2L-2$, and the size of its trajectory matrix $\widetilde\bfX$ is $L$ by $N+L-1$ (instead of $N-L+1$ for the non-extended trajectory matrix).

For the extended series, Algorithm~\ref{alg:weightedSVD} with weights $m_{i,j}=\calT \tsI$ is applied to $\widetilde\bfX$, where the series $\tsI$ has ones in the place of the series $\tsX$ and zeroes in positions of gaps, i.e.
\begin{equation*}
m_{i,j} = \begin{cases}
1 & 1 \le i+j-L \le N, \\
0 & \text{otherwise.}
\end{cases}
\end{equation*}

\begin{algorithm}[Extended Cadzow iterations]\label{alg:ECIt}
	\textbf{Input}: Time series $\tsX$, window length $L$, rank $r$,
	stop criteria STOP1 for outer iterations and STOP2 for inner iterations,
	left and right extension values $\tsL_{L-1}$ and $\tsR_{L-1}$.
	
	\textbf{Result}:
	Approximation $\widehat\tsS$ of time series $\tsX$ by finite-rank series of rank $r$.
	
	\begin{enumerate}
		\item
		$\widetilde\bfY_0 = \calT \widetilde\tsX$, where $\widetilde\tsX=(\tsL_{L-1}, \tsX, \tsR_{L-1})$, $k=0$.
		\item
		Obtain $\widehat\bfZ$ using Algorithm ~\ref{alg:weightedSVD} applied to $\widetilde\bfY_k$ for estimation of $\Pi_{\calM_r} \widetilde\bfY_{k}$ with stop criterion STOP2.
		\item
		$\widetilde\bfY_{k+1} = \Pi_\calH  \widehat\bfZ$, $k\leftarrow k+1$.
		\item
		Construct $\bfY_k$ consisting of the columns of the matrix $\widetilde\bfY_{k}$, from $L$-th to $N$-th ones. If STOP1, then $\widehat\tsS = \calT^{-1} \bfY_k$; else go to 2.
	\end{enumerate}
\end{algorithm}


\subsection{Oblique Cadzow iterations}
\label{sec:ObliqueCadzow}
Algorithms considered in this section generalize the conventional Cadzow algorithm based on the Euclidean inner product to the use of an oblique inner product given by a matrix $\bfC$.
These algorithms can be applied if the conditions of Proposition~\ref{prop:projS} hold.

\begin{algorithm}[Oblique Cadzow iterations]
	\label{alg:obliqueCadzow}
	\textbf{Input}: Time series $\tsX$, window length $L$, rank $r$, matrix $\bfC=\diag(c_1,\ldots, c_K)$, where $K=N-L+1$,
	stop criteria STOP1.
	
	\textbf{Result}:
	Approximation $\widehat\tsS$ of time series $\tsX$ by finite-rank series of rank $r$.
	
	\begin{enumerate}
		\item
		$\bfY_0 = \calT \tsX$, $k=0$.
		\item
		$\bfY_{k+1} = \Pi_\calH  \Pi_{\calM_r} \bfY_{k}$, $k\leftarrow k+1$, where
		$\Pi_{\calM_r}$ is given by \eqref{eq:PiMr}.
		\item
		If STOP1, then $\widehat\tsS = \calT^{-1} \bfY_k$; else go to 2.
	\end{enumerate}
\end{algorithm}

To solve the problem \eqref{L-rank_task} of approximation of time series with equal weights $q_i$, a proper matrix $\bfC$ should be chosen. It is found that there is no such full-rank matrix; therefore, a few variants providing approximately equal weights are considered below.

\subsubsection{Cadzow($\alpha$) iterations}
\label{sec:cadzow_alpha}
The following lemma describes a case, when the conditions of Proposition~\ref{prop:equiv_tasks} are fulfilled and therefore
the problem \eqref{L-rank_task} with equal weights $q_i$ is equivalent to the problem \eqref{rank_task}.

\begin{lemma}[\cite{Gillard2014}]
	\label{zhiglemma}
	Let $\tsX \in \sfX_N$, $\bfX = \calT(\tsX) \in \sfR^{L \times K}$. If $h = N/L$ is integer, then for $q_i\equiv 1$ we have $\|\tsX\|^2_q = \|\bfX\|^2_{2, \bfC}$, where $\bfC=\diag(c_1,\ldots,c_K)$ with diagonal elements
	\begin{equation*}
	c_k = \begin{cases}
	1, & \text{if} \; k = jL+1 \;\text{for some} \; j = 0, \ldots, h-1, \\
	0, & \text{otherwise}.
	\end{cases}
	\end{equation*}
\end{lemma}

This approach has an essential drawback. Since zeroes are placed at the diagonal of the diagonal matrix $\bfC$, $\bfC$ has rank $h$, which is considerably smaller than $K$. The change of the diagonal zeroes to some small $\alpha$ is suggested in \cite{Gillard2014} to improve rank-deficiency.

Let
\begin{multline}\label{zhigweights}
c_k = c_k(\alpha) =\\ \begin{cases}
1, & \text{if} \; k = jL+1 \; \text{for some} \; j = 0, \ldots, h-1, \\
\alpha, & \text{otherwise.}
\end{cases}
\end{multline}
Then the matrix $\bfC(\alpha)=\diag(c_1(\alpha),\ldots,c_K(\alpha))$ with the diagonal given in \eqref{zhigweights} is of full rank.
However, the corresponding series weights are not equal.

Let Cadzow($\alpha$) denote the iterations performed by Algorithm~\ref{alg:obliqueCadzow} with the diagonal matrix $\bfC=\bfC(\alpha)$.
	Note that for $\alpha=1$ the matrix $\bfC(\alpha)$ is the identity matrix and the Cadzow($\alpha$) iterations coincide with the conventional Cadzow iterations.

\paragraph*{Degenerate case $\alpha=0$.}

Equality \eqref{sk_mlk} provides the form of a matrix $\bfM$ to obtain $\|\cdot\|_{1, \bfM} = \|\cdot\|_{2, \bfC}$ in the case $\alpha = 0$:
\begin{equation}
\label{eq:degenerateM}
\bfM = \begin{pmatrix}
1 & 0 & 0 & \cdots & 0 & 1 & 0 & \cdots & \cdots & 1 \\
1 & 0 & 0 & \cdots & 0 & 1 & 0 & \cdots & \cdots & 1 \\
\vdots & \vdots & \vdots & \cdots & \vdots & \vdots & \vdots & \cdots & \cdots & 1 \\
1 & 0 & 0 & \cdots & 0 & 1 & 0 & \cdots & \cdots & 1
\end{pmatrix}.
\end{equation}

\begin{remark}
	The optimization problem \eqref{rank_task} with the matrix $\bfM$ given in \eqref{eq:degenerateM} corresponds to the search of an arbitrary (not necessary Hankel) matrix of rank not larger than $r$, which is closest in the Frobenius norm to the matrix
	\be
	\label{eq:traj_noinersect}
	\begin{pmatrix}
		x_1&x_{L+1}&\cdots&x_{K}\\
		\vdots&\vdots&\cdots&\vdots\\
		x_L&x_{2L}&\cdots&x_N
	\end{pmatrix}.
	\ee
	This problem is quite different from the problem \eqref{L-rank_task} of approximation by finite-rank series. Therefore, the Cadzow($0$) algorithm does not solve the problem \eqref{L-rank_task}.
\end{remark}

\subsubsection{Cadzow-$\widehat{\bfC}$ iterations}
\label{sec:cadzow_hat}
Let us correct the rank-deficiency of $\bfC(0)$ by another way.

To obtain equal series weights $q_i\equiv 1$ in \eqref{L-rank_task}, we
should choose the weight matrix $\bfM$ in \eqref{rank_task} with weights $m_{i,j}$ defined in \eqref{Mw}.
Generally, there is no a matrix $\bfC$ providing the equivalent norm
$\|\cdot\|_{2,\bfC}=\|\cdot\|_{1,\bfM}$, since the matrix $\bfC$ should
be diagonal and therefore the matrix $\bfM$ should have columns consisting of equal elements (see Proposition~\ref{prop:equiv_tasks}).

To obtain approximately equal weights, the following approach is suggested.
Consider the set $\sfZ \subset \sfR^{L \times K}$ of matrices with columns consisting of equal elements and find $\what{\bfM}$ such that
\begin{equation*}
\|\bfM - \widehat\bfM\| \to \min_{\widehat\bfM \in \sfZ},
\end{equation*}
where $\|\cdot\|$ is the Frobenius norm.

The solution $\widehat\bfM$ is evidently constructed as the averaging of the matrix $\bfM$ by columns. As a result, the resultant matrix $\widehat\bfC$ such that $\|\cdot\|_{2,\widehat\bfC}=\|\cdot\|_{1,\widehat\bfM}$ has the form $\widehat\bfC = \text{diag}(\hat c_1, \ldots, \hat c_K)$, where
\begin{equation}\label{my_s}
\hat c_k = \frac{1}{L}\sum_{l=1}^L m_{l, k}.
\end{equation}


We call Algorithm~\ref{alg:obliqueCadzow} with the matrix $\bfC=\widehat\bfC$ Cadzow-$\widehat\bfC$ iterations.

\subsubsection{Weights $q_i$ in \eqref{L-rank_task} produced by the algorithms}
Since the norm $\|\cdot\|_{2, \bfC}$ with $\bfC(\alpha)$ or $\widehat\bfC$ in place of $\bfC$ does not correspond to equal series weights,
let us find $q_i(\alpha)$ and $\hat{q}_i$ from the equalities
$\|\bfY\|_{2,\widehat\bfC}=\|\tsY\|_{\hat{q}}$ and $\|\bfY\|_{2,\bfC(\alpha)}=\|\tsY\|_{q(\alpha)}$.
Formulas for calculation are provided in Proposition~\ref{prop:equiv_tasks}.

The following statements are valid.

\begin{proposition}\label{prop:zhigconseq}
	Let $h = N/L$ be integer, $\bfC(\alpha) = \diag(c_1(\alpha), \ldots, c_K(\alpha))$, where $c_i(\alpha)$ are given in \eqref{zhigweights}, $0 \le \alpha \le 1$. Then the weights $q_i(\alpha)$ have the form
	\begin{equation*}
	q_i (\alpha) = \begin{cases}
	1 + (i - 1) \alpha & \text{для $i = 1, \ldots, L-1,$}\\
	1 + (L - 1) \alpha & \text{для $i = L, \ldots, K-1,$}\\
	1 + (N - i) \alpha & \text{для $i = K, \ldots, N.$}
	\end{cases}
	\end{equation*}
\end{proposition}
\begin{proof}
	The proof is a straightforward consequence of Proposition~\ref{prop:equiv_tasks}.
\end{proof}

To illustrate the form of the weights $\hat{q_i}$, let us formulate propositions with simplifying conditions.
\begin{proposition} \label{myweightstat}
	Let $N \ge 3(L-1)$. Then the diagonal matrix weights $\hat c_k$ defined in \eqref{my_s} are equal to
	\begin{equation*}
	\hat c_k = \begin{cases}
	\frac{1}{L}\left(\frac{k}{L} + \sum_{j=k}^{L-1} \frac{1}{j} \right),& 1 \le k \le L-1, \\
	1/L, & L \le k \le K - L + 1,\\
	\hat c_{K - k + 1}, & K - L + 2 \le k \le K.
 	          \end{cases}
	\end{equation*}
\end{proposition}

\begin{proof}
	To prove the proposition, it is sufficient to substitute $m_{l,k}$ defined in \eqref{Mw} to \eqref{my_s}.
\end{proof}

\begin{proposition} \label{myserweightstat}
	Let $N \ge 4(L-1)$. Define
	\begin{equation*}
	\hat{u}_i = \begin{cases}
	\frac{i(i+1)}{2 L^2} + \frac{i}{L}(1 + H_{L-1} - H_i), &1 \le i \le L-1, \\
	1 + \frac{2iL-i-i^2}{2L^2} + \frac{L-i}{L}(H_{L-1} - H_{i - L}), & L \le i \le 2L-1,
	\end{cases}
	\end{equation*}
	where $H_0 = 0$, and $H_i = \sum_{j=1}^i 1/j$ is the $i$-th harmonic number.
	 Then the weights $\hat{q}_i$ have the form:
	\begin{equation*}
	\hat{q}_i = \begin{cases}
	\hat{u}_i, &1 \le i \le 2L-1, \\
		1, &2L \le i \le N-2L+1,\\
\hat{u}_{N-i+1}, &N-2L+2 \le i \le N. \\
	\end{cases}
	\end{equation*}
\end{proposition}

\begin{proof}
	For $1 \le i \le L-1$, we have
	\begin{multline*}
	\hat{q}_i = \sum_{j=1}^i \hat{c}_j = \sum_{j=1}^i \frac{1}{L}\left(\frac{j}{L} + \sum_{k=j}^{L-1}\frac{1}{k}\right)\! =
	\frac{i(i+1)}{2L^2} + \\ + \frac{1}{L} \sum_{k = 1}^{L-1} \sum_{j=1}^{\min(k,i)} \frac{1}{k} = \frac{i(i+1)}{2L^2}+\frac{1}{L} \sum_{k = 1}^{L-1} \frac{\min(k,i)}{k} = \\ = \frac{i(i+1)}{2 L^2} + \frac{i}{L}(1 + H_{L-1} - H_i).
	\end{multline*}
	For $L \le i \le 2L-1$, changing the order of summation, we obtain
	\begin{multline*}
	\hat{q}_i = \sum_{j = 1}^L \hat{c}_{i-L+j} = \sum_{j = i - L + 1}^{L - 1} \hat{c}_j + \frac{i - L + 1}{L} =\\
	=\frac{i - L + 1}{L} + \frac{1}{L^2} \sum_{j = i - L + 1}^{L-1}j + \frac{1}{L} \sum_{j = i-L + 1}^{L-1} \sum_{k=j}^{L-1}\frac{1}{k} =\\
	=\frac{i - L + 1}{L} + \frac{2iL - i - i^2}{2L^2} + \frac{1}{L} \sum_{k = i - L + 1}^{L - 1} \sum_{j = i - L + 1}^k \frac{1}{k} =\\
	=1 + \frac{2iL-i-i^2}{2L^2} + \frac{L-i}{L}(H_{L-1} - H_{i - L}).
	\end{multline*}
	The weights $\hat{q}_i$ for $N-2L+2 \le i \le N$ are calculated by symmetry. The center series weights are evidently equal to 1.
\end{proof}

\begin{figure}[!hhh]
\begin{center}
		\includegraphics[width =0.7\columnwidth]{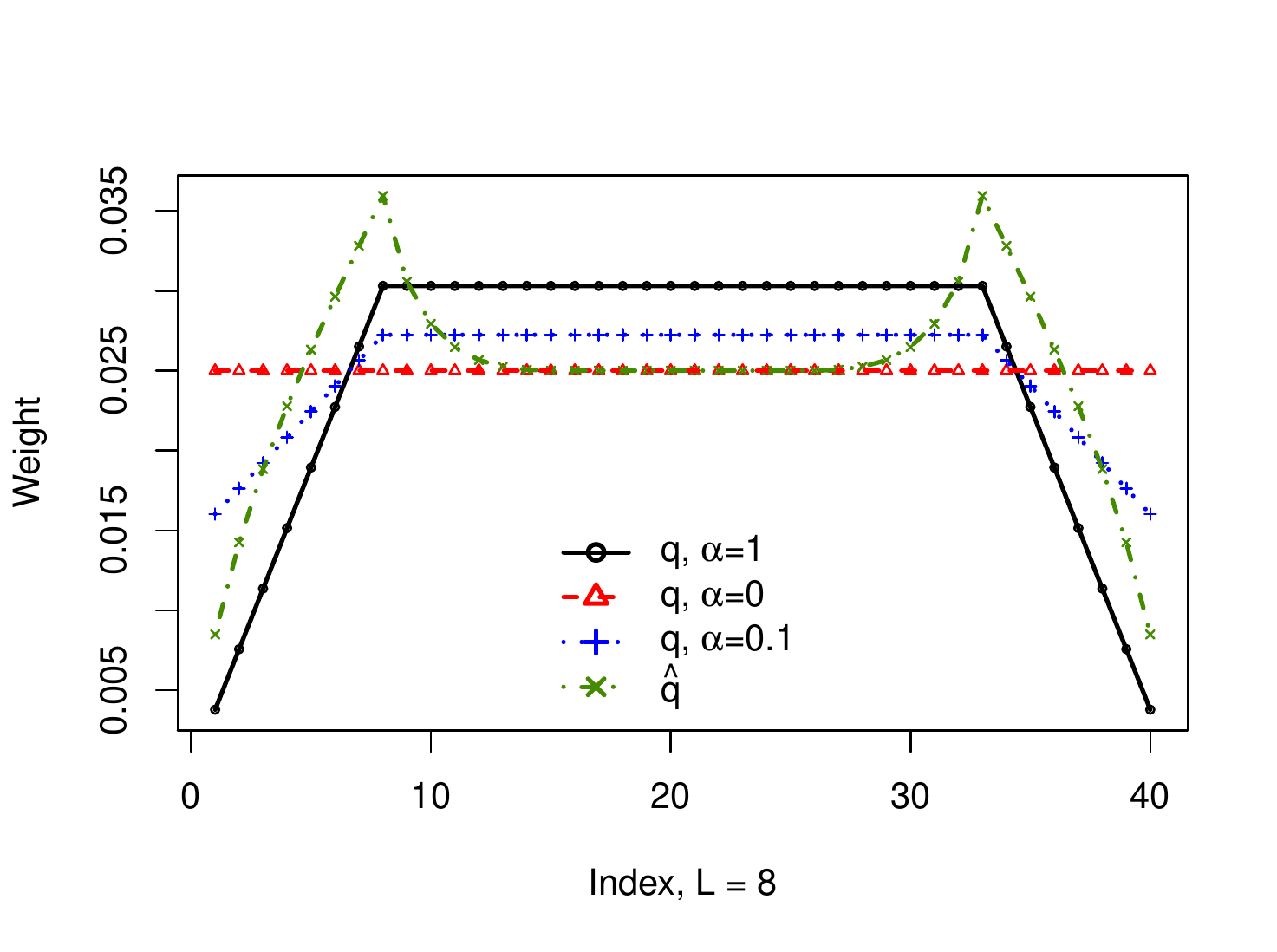}
\end{center}
\caption{Normalized series weights $q_i$ corresponding to $\bfC(\alpha)$ and $\widehat\bfC$.}\label{img_weights}
\end{figure}

Let us normalize series weights so that their sums equal 1.
The normalized weights $q_i(\alpha)$, for $\alpha = 1$ (the conventional Cadzow iterations), $\alpha = 0$ (equal $q_i$), $\alpha = 0.1$,
and $\hat{q}_i$ for $N = 40$, $L = 8$, are shown in Figure~\ref{img_weights}.
	
	\subsection{Comments to algorithms. Comparison}
\label{sec:comments}
	
	Let us comment and compare the following methods: the Weighted Cadzow iterations (Algorithm \ref{alg:WCIt}), the Extended Cadzow iterations (Algorithm \ref{alg:ECIt}), the Cadzow($\alpha$) iterations, $0< \alpha \leq 1$, coinciding with the conventional Cadzow iterations if $\alpha=1$,
	and finally the Cadzow-$\widehat\bfC$ iterations (Algorithm \ref{alg:obliqueCadzow}).
	Note that the window length $L$ is a parameter for each of the considered methods.
	
	\begin{itemize}
		\item \textit{Theoretical convergence.}
		Theorem \ref{th:converg} provides conditions for the existence of a subsequence, which converges to a matrix from $\calM_r \cap \calH$. This theorem is applicable directly to Algorithm \ref{alg:obliqueCadzow} if all weights are positive and to Algorithm \ref{alg:WCIt} if to suppose that the weighted projection to $\calM_r$ can be calculated with no error. It is easy to extend Theorem \ref{th:converg} to be applicable to Algorithm \ref{alg:ECIt} where the weights for added values are zero, if to consider the sequence  $\bfY_k$ instead of $\widetilde\bfY_k$.
		\item \textit{Convergence in practice.} Although the theory says about the existence of converging subsequences, the convergence of the constructed sequences took place in all the training examples.
		\item \textit{Comparison by accuracy.}
		The methods are iterative, and convergence to the global minimum in the corresponding least-squares problem does not necessarily take place. Therefore, different algorithms corresponding to the same weights can yield different approximations. Hence, the comparison of the algorithms by the approximation accuracy makes sense.		
\item \textit{Signal estimation and series approximation.}
		The proposed methods can be considered as both approximation methods of the original series by finite-rank series and weighted least-squares methods for signal estimation. Note that generally the approximation quality can contradict to the estimation accuracy due to possible over-fitting.
		\item \textit{Algorithms and series weights.}
		The Weighted Cadzow and Extended Cadzow methods try to solve the problem \eqref{L-rank_task} with equal weights $q_i$. The other methods work with weights with different levels of non-uniformity.
		\item \textit{Algorithms and computational costs.}
		All suggested algorithms are iterative. However, each outer iteration in the Weighted Cadzow and Extended Cadzow algorithms has a step with inner iterations. Therefore, these algorithms are very time-consuming. The other algorithms do not contain inner iterations; moreover, they have similar computational costs of one iteration and can be compared by the number of iterations.
		Computational complexity is described by both complexity of one iteration and the number of iterations. Evidently, the necessary number of iterations is determined by the convergence rate.
\item \textit{Fast implementation.}
There is a very fast implementation of iterations of the Cadzow algorithm suggested in \cite{Korobeynikov2010} and extended in \cite{Golyandina.etal2015}. However, it can be shown that the same implementation approach can be applied to the Cadzow($\alpha$) and Cadzow-$\widehat{\bfC}$ algorithms. Therefore, fast implementations of these algorithms still can be compared by the number of iterations.
		\item \textit{Use of the first iteration for signal estimation.}
		One iteration of the Cadzow iterations is exactly the well-known Singular Spectrum Analysis (SSA) method, which can solve a significantly wider range of tasks than the iterative method does. By analogy, together with the limiting series, we are interested in the signal estimation by means of the first iteration of the considered algorithms. In a sense, each iterative method produces a modification of SSA. The first iteration is generally not of finite rank; however, it has low computational complexity and can provide sufficient accuracy.
		\item \textit{Separability, the first iteration and the convergence rate.}
		Separability of a signal, which is an important concept of the SSA method, means the ability of a method to (approximately) separate the signal from a residual. From the viewpoint of the iterative methods,  the separability quality is closely related to the accuracy of the first iteration of the method. On the other hand, we can expect that the accuracy of the first iteration is connected with the method's convergence rate. Therefore, the separability accuracy is connected with the convergence rate of iterative methods.
		\item \textit{Separability and choice of parameters.}
		The connection between the separability and the window length $L$ is well studied for the SSA method, see \cite{Golyandina2010}. In particular, optimal window lengths are close to half of the series length. A small window length $L$ provides poor separability. We can expect that this is valid for the other algorithms too.
The Cadzow($\alpha$) method has an additional parameter $\alpha$.
Influence of the parameter $\alpha$ on separability in the class of Cadzow($\alpha$) iterations is investigated in Appendix~\ref{sec:app}.
 The studied example of separability of a sine-wave signal from a constant residual shows that small values of $\alpha$ provide poor separability.
        \item \textit{Equal series weights and choice of parameters.}
        Let us consider the dependence of series weights produced by the Cadzow($\alpha$) algorithm on the window length $L$ or $\alpha$. Proposition~\ref{prop:zhigconseq} shows that more uniform weights are achieved for small $L$ and small $\alpha$. This is exactly
        the case corresponding to poor separability.
 		\item \textit{Equal series weights and accuracy of signal estimation.}
 Thus, the weights, which are close to equal ones, correspond to algorithms, which either have a time-consuming iteration step with inner iterations or are slowly convergent; therefore such algorithms have high computational complexity.
There are no theoretical results about the behavior of the estimation accuracy in dependence on algorithms and their parameters. However, the numerical study shows that the best accuracy is achieved in the algorithms corresponding to the weights, which are equal or almost equal.	 \end{itemize}
	
	\begin{remark}
		\label{rem:adjust}
		The adjustment $\calA$, which is suggested in Section~\ref{sec:ts} for improvement of estimates, can be applied to the resultant signal estimation $\widehat\tsS$ for any considered algorithm. The inner product used in Remark~\ref{rem:adj} for definition of $\calA$ is the standard Euclidean inner product not depending on the weight matrix $\bfM$ used in the algorithms, since this norm $\|\cdot\|$ corresponds to the problem \eqref{L-rank_task} with equal weights $q_i$. We will call the algorithms with the adjustment $\calA$ \emph{adjusted algorithms}. For example, the result of the $k$-th iteration of the Cadzow iterations can be expressed as $\widehat\tsS_k = \calT^{-1}(\Pi_\calH \Pi_{\calM_r})^k \calT \tsX$. Then the result of the $k$-th iteration of the adjusted Cadzow iterative method is $\widehat\tsS_k^*=\calA(\widehat\tsS_k)$.
	\end{remark}

\section{Numerical comparison}
\label{sec:simul}
Let us carry out numerical experiments for analysis of the performance of the considered methods. Comparison of the methods was performed on several examples, with a sine-wave signal and an exponentially-modulated sine-wave signal.
Since the obtained comparison results are very similar, only the results for a sine-wave signal are presented.

Suppose that the signal $\tsS = (s_{1}, \ldots, s_N)$ of length $N = 40$ and rank $r=2$ has the form:
\begin{equation}
\label{eq:signal}
s_{k} = 5\sin{\frac{2 \pi k}{6}}, \quad k = 1, \ldots, N,
\end{equation}
and the series $\tsX = \tsS + \tsN$ is observed, where  $\tsN$ is Gaussian white noise with mean equal to $0$ and variance equal to $1$. Accuracy of a signal estimate $\widehat\tsS$ is measured as the root mean-square error (RMSE) using 1000 simulations. Comparison is performed on the same simulated samples. It was checked that the stated comparison results are significant at the 5\% level of significance.

\smallskip
\textit{Convergence rate and accuracy.} We start with the investigation of the Cadzow-$\widehat\bfC$ method and the Cadzow($\alpha$) methods for several values of $\alpha$, since they have not internal iterations and therefore their computational costs can be compared by the number of external iterations. These methods use an oblique SVD; the Cadzow($1$) method is the conventional Cadzow method. Figure~\ref{img_cadzowspeed2} shows the rate of convergence for $\alpha = 0.1$ and $\alpha = 1$ and for two different window lengths $L$. The RMSE values are depicted versus the number of performed iterations.

\begin{figure}[!hhh]
		\includegraphics[width = \columnwidth]{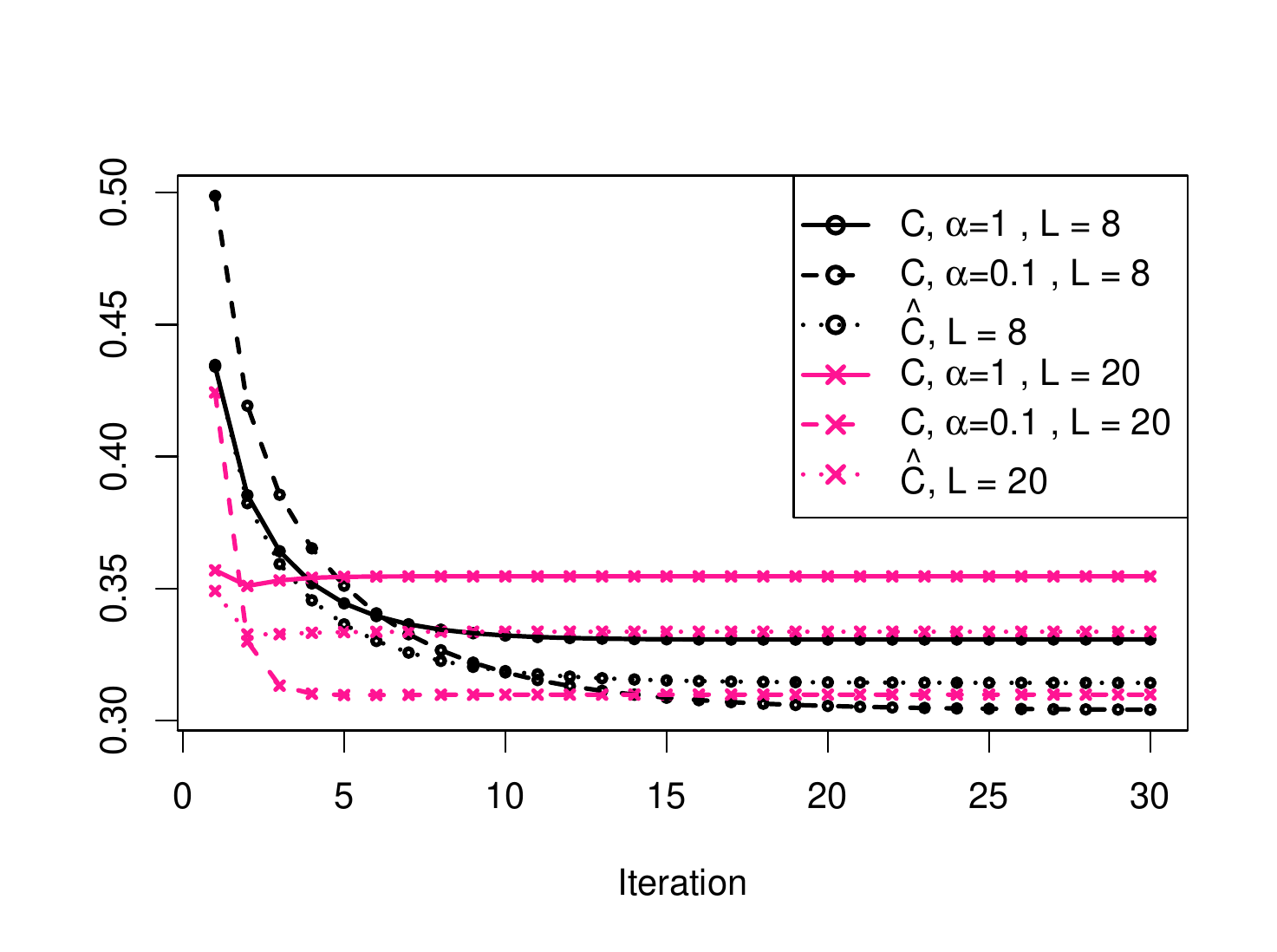}
		\caption{The RMSE of the signal estimate depending on the number of iterations,
        $\sigma=1$.}
		\label{img_cadzowspeed2}
\end{figure}

One can see that a method with a smaller limit error is the one with a slower convergence rate. For parameters involved to the simulations, the Cadzow($0.1$) method with the window length $L=8$ has the smallest limit error. At the same time, these values of parameters correspond to both the slowest convergence and the most uniform weights.

Note that the limit errors do not differ strongly, they change from 0.31 ($\alpha=0.1$, $L=8$) in the best case to 0.35 ($\alpha=1$, $L=20$) in the worst case. However, the error equal to 0.35 is achieved at the first iteration in the worst case, while it takes 4--5 iterations to achieve the error 0.35 in the best case.

\begin{figure}[!hhh]
		\begin{center}
	\includegraphics[width = 0.7\columnwidth]{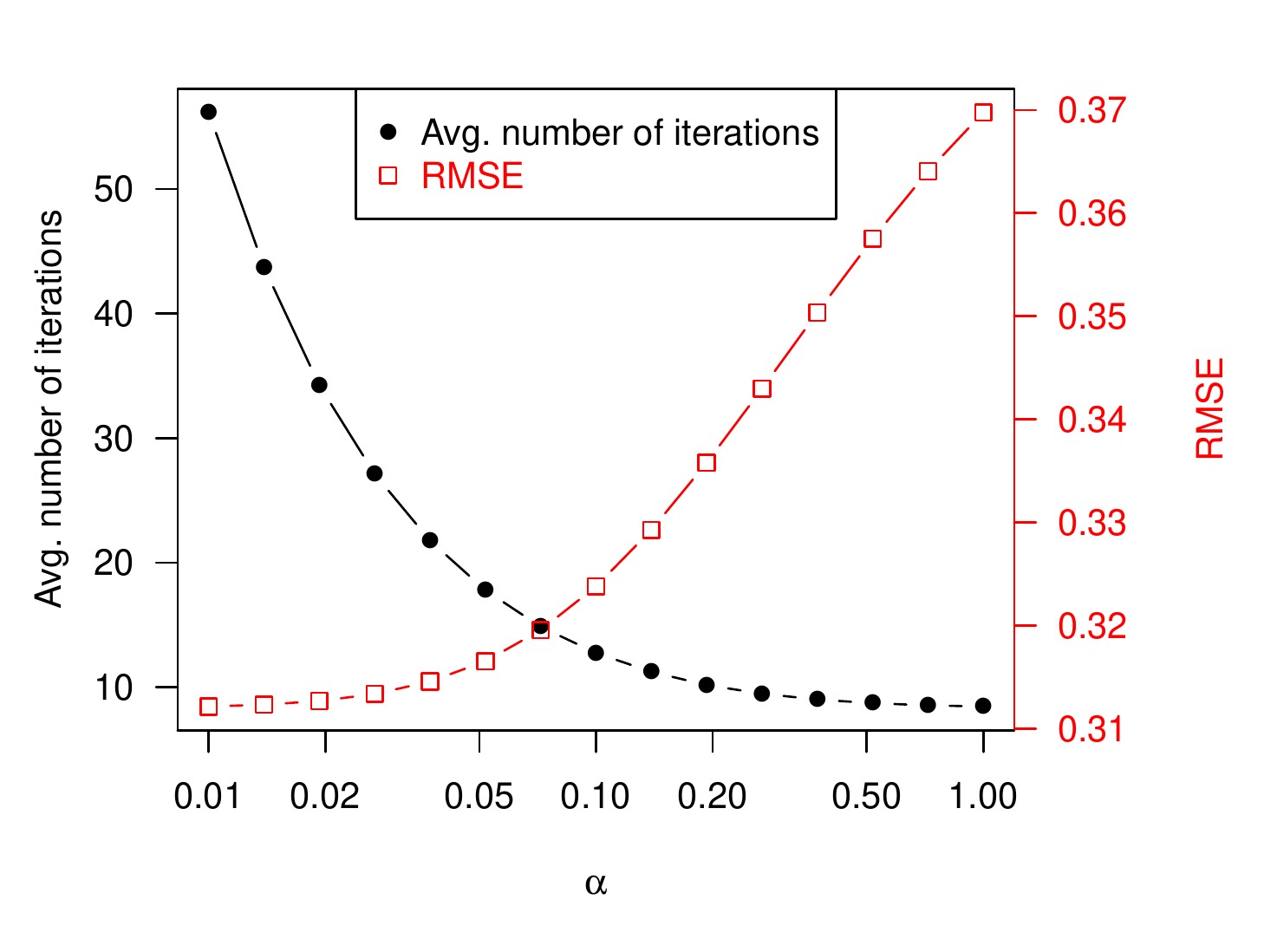}
	\caption{The RMSE and the average number of iterations depending on $\alpha$ (log-scale), $L=20$, $\sigma=1$.}
	\label{img_2axis}
		\end{center}
\end{figure}

\smallskip
\textit{Accuracy vs number of iterations for Cadzow($\alpha$).}
The same signal \eqref{eq:signal} was taken to investigate how the RMSE and the convergence rate depend on $\alpha$ for the Cadzow($\alpha$) algorithms with $L = 20$. The following STOP1 criterion was taken: $\frac{\|\calT^{-1}(\bfY_k) - \calT^{-1}(\bfY_{k + 1})\|^2}{N} < 10^{-8}$.

Figure~\ref{img_2axis} shows that smaller $\alpha$ leads to more accurate estimates of the signal, but increases their computational costs.
This general rule sometimes does not work for very small values of the parameter $\alpha$, see Figure~\ref{img_2axis-3}, where
the noise standard deviation was increased from 1 to 3. One can see that for $\alpha$ smaller than 0.1 the dependence of the estimation errors on $\alpha$ changes.
It seems that the threshold $\alpha$, which corresponds to the change of the accuracy behaviour, depends on the 1-iteration separability of the signal from noise.
Indeed, as we can expect, for small $\alpha$ the separability quality is poor (see an example in Appendix~\ref{sec:app}).

\begin{figure}[!hhh]
		\begin{center}
	\includegraphics[width = 0.7\columnwidth]{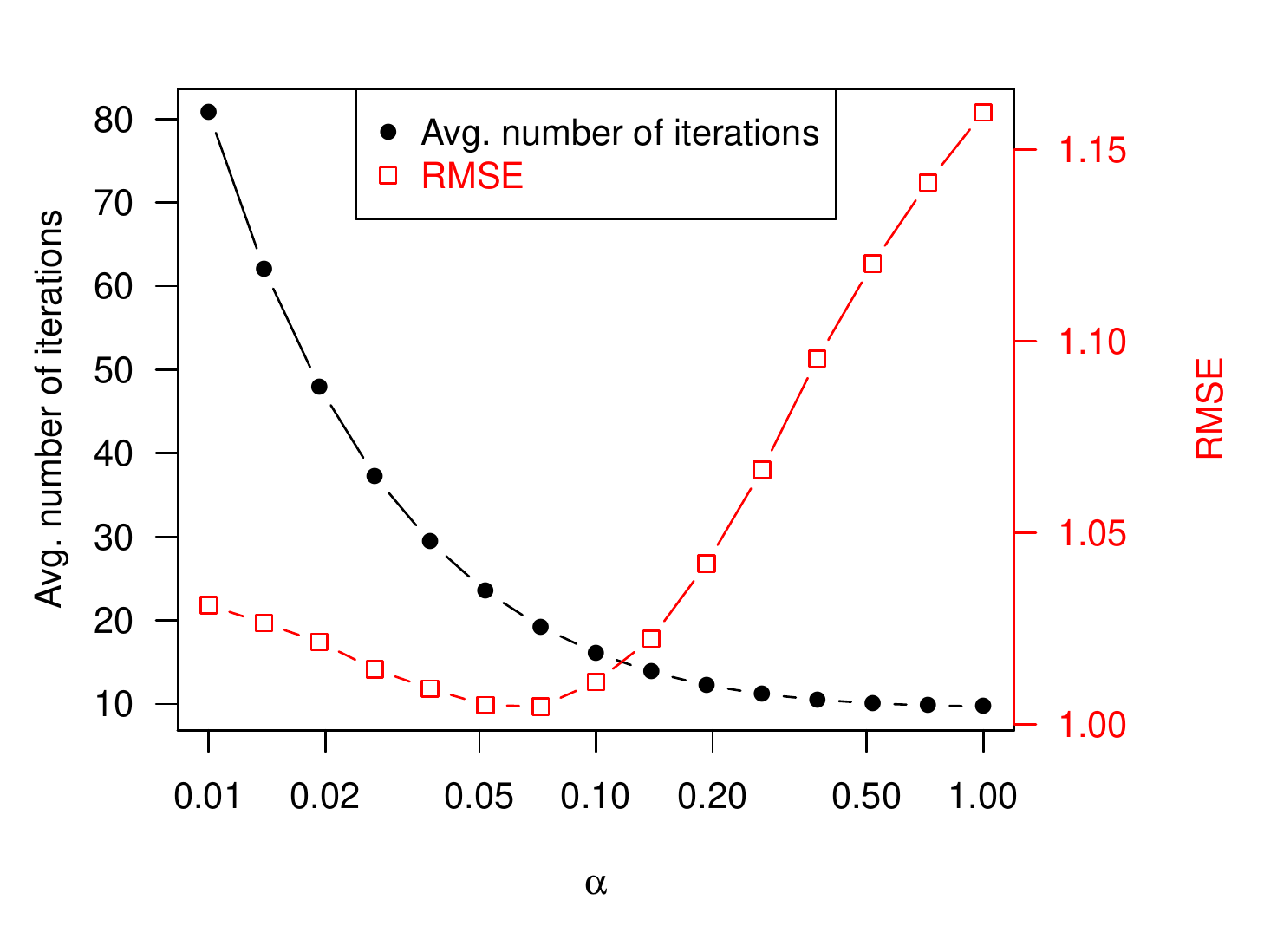}
	\caption{The RMSE and the average number of iterations depending on $\alpha$ (log-scale), $L=20$, $\sigma=3$.}
	\label{img_2axis-3}
		\end{center}
\end{figure}

\smallskip
\textit{Comparison by accuracy at the first iteration and in the limit.} Let us now involve the Extended and Weighted Cadzow iterations and examine the spreading of the estimation errors along the series. The maximal number of iterations equal to 100 is taken for the stop criterion STOP1 (this choice yields the error close to the limiting value); the stop criterion STOP2 for inner iterations is as follows:
$\frac{\|\bfY_k - \bfY_{k+1}\|^2}{LK} < 10^{-4}$. The initial left and right extended values $\tsL_{L-1}$ and $\tsR_{L-1}$ in the Extended Cadzow iterations are obtained using the vector SSA-forecasting method \cite[Section 2.3.1]{Golyandina.etal2001}.

Figures~\ref{fig:s1_it1}~and~\ref{fig:s1_it100} show the dependence of the RMSE on numbers of the series points. Figure~\ref{fig:s1_it1} shows the errors at the first iteration, Figure~\ref{fig:s1_it100} shows the errors at the 100-th iteration. It is clearly seen that the Extended Cadzow method is the most precise in both cases. The Cadzow($1$) and Cadzow-$\widehat\bfC$ methods are the best at the first iteration among the set of methods without inner iterations. The best method in the limit (after 100-th iteration, the errors do not change significantly further) is the Cadzow($0.1$) method; this is not surprising according to Figure~\ref{img_cadzowspeed2}.

\begin{figure}[!hhh]
	\begin{center}
		\includegraphics[width = 0.7\columnwidth]{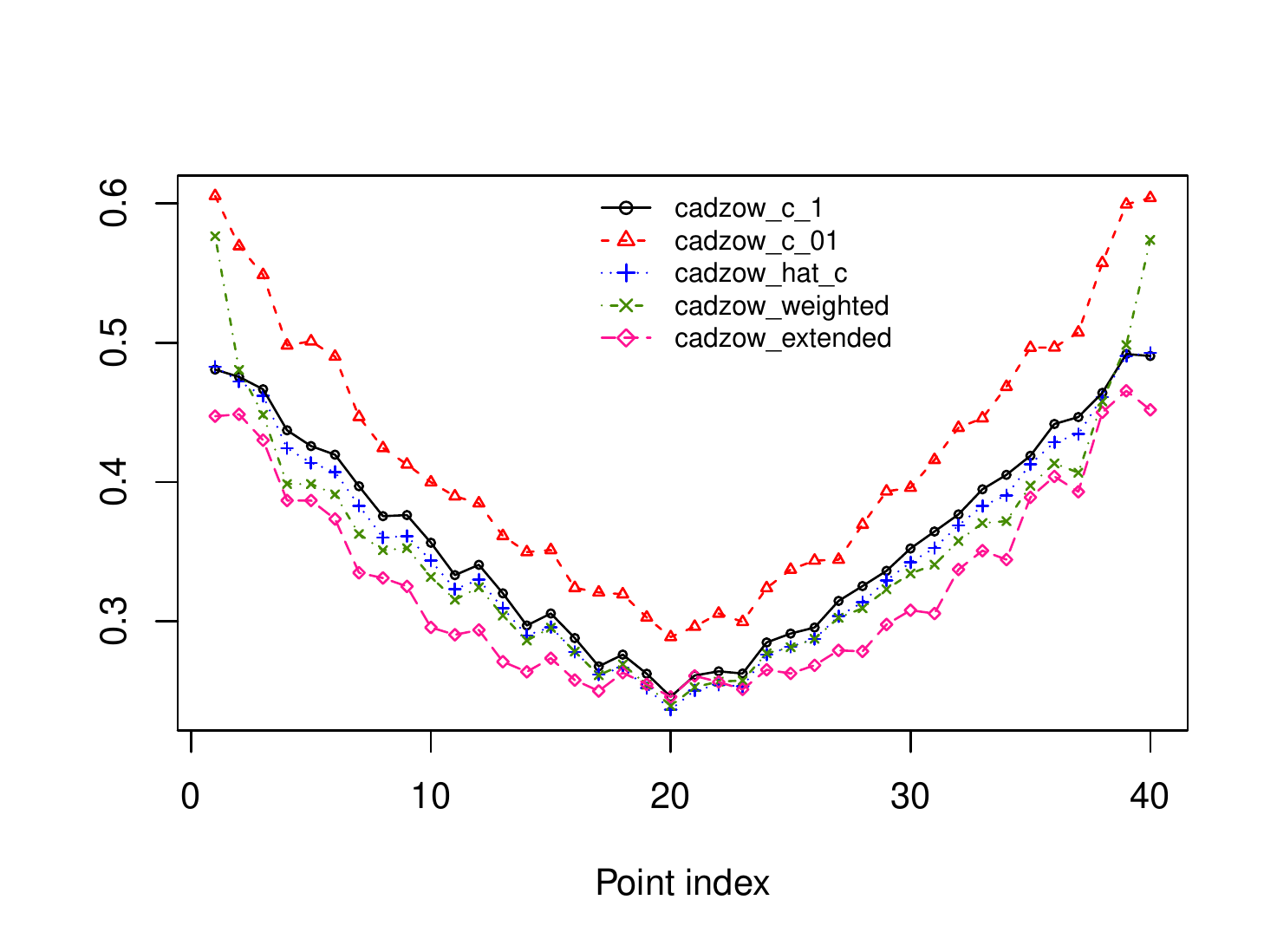}
		\caption{The RMSE of signal estimates at each series point; iteration 1; $L=20$, $\sigma=1$.}
		\label{fig:s1_it1}
	\end{center}
\end{figure}

\begin{figure}[!hhh]
	\begin{center}
		\includegraphics[width = 0.7\columnwidth]{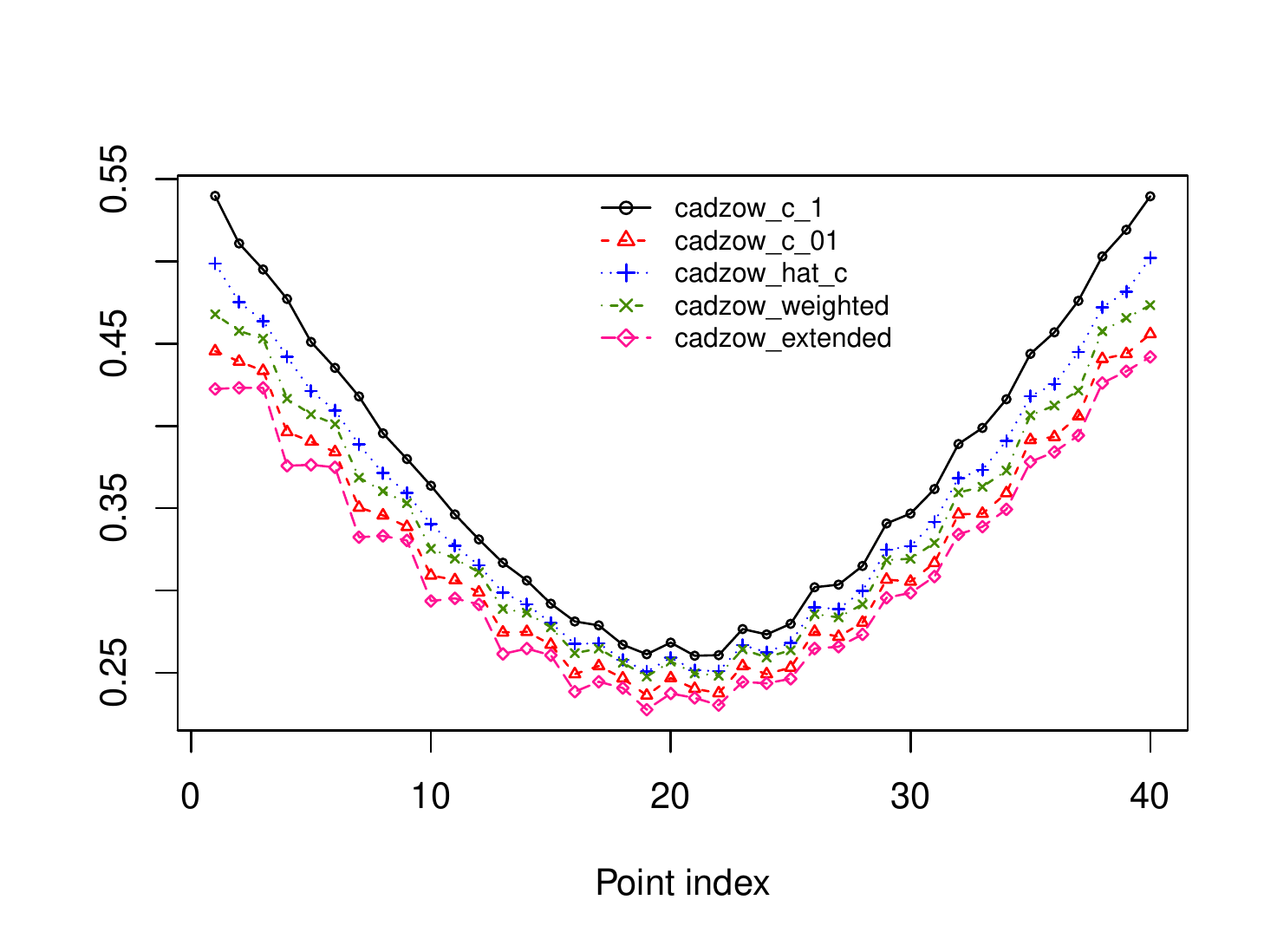}
		\caption{The RMSE of signal estimate at each series point; iteration 100; $L=20$, $\sigma=1$.}
		\label{fig:s1_it100}
	\end{center}
\end{figure}

\smallskip
\textit{Errors for signal and original series approximations.} Since we used the least-squares method for estimation of the signal $\tsS$, consider Table~\ref{fintable} which shows the RMSE for $\tilde \tsS$ as an estimate of $\tsS$ (i.e., the signal estimation errors) and the RMSE for $\tilde \tsS$ as an estimate of the original series $\tsX$ (i.e., the series approximation errors). Here $k$ is the number of iterations, $L=20$. Table~\ref{fintable} confirms the conclusions about comparison of the methods by accuracy of signal estimation. Also it is seen that the quality of original series approximation does not always correspond with the quality of signal estimation. For example, overfitting is clearly present for the Cadzow($0.1$) iterations at the first iteration. However, the methods are ordered identically by errors of series approximation and signal estimation in the limit. This means that minimization of the error of approximation likely yields minimization of the error of signal reconstruction.
The same ordering of the errors is very important for practice, since for real-life data we can choose a better method and its parameters by smaller approximation errors.
Certainly, a proper rank should be set before the comparison of the methods.

\begin{table*}
		\caption{Comparison of methods by the RMSE, $L = 20$, $\sigma=1$, for the signal \eqref{eq:signal}.}\label{fintable}

		\begin{tabular*}{\textwidth}{@{\extracolsep{\fill}}lcccc}
			\hline
			Method & $\tsS$, $k = 1$ & $\tsX$, $k = 1$ & $\tsS$, $k = 100$ & $\tsX$, $k = 100$  \\
			\hline
			Cadzow, $\alpha = 1$ & 0.3758 & 0.9195 & 0.3782 & 0.9664 \\
			\hline
			Cadzow, $\alpha = 0.1$ & 0.4329 & 0.7040 & 0.3311 & 0.9506 \\
			\hline
			Cadzow $\hat{\bfC}$ & 0.3655 & 0.8925 & 0.3559 & 0.9583 \\
			\hline
			Weighted Cadzow & 0.3644 & 0.8891 & 0.3455 & 0.9549 \\
			\hline
			Extended Cadzow & 0.3361 & 0.9030 & 0.3189 & 0.9471 \\
			\hline
		\end{tabular*}
\end{table*}

\begin{table*}
	\begin{center}
		\caption{Comparison of adjusted methods by the RMSE, $L = 20$, $\sigma=1$, for the signal \eqref{eq:signal}.}\label{fintable_improved}
		\begin{tabular*}{\textwidth}{@{\extracolsep{\fill}}lcccc}
			\hline
			Method & $\tsS$, $k = 1$ & $\tsX$, $k = 1$ & $\tsS$, $k = 100$ & $\tsX$, $k = 100$  \\
			\hline
			Cadzow, $\alpha = 1$ & 0.3714 & 0.9175 & 0.3667 & 0.9622 \\
			\hline
			Cadzow, $\alpha = 0.1$ & 0.4385 & 0.7023 & 0.3276 & 0.9493 \\
			\hline
			Cadzow $\hat{\bfC}$ & 0.3626 & 0.8909 & 0.3478 & 0.9555 \\
			\hline
			Weighted Cadzow & 0.3640 & 0.8883 & 0.3380 & 0.9523 \\
			\hline
			Extended Cadzow & 0.3370 & 0.9030 & 0.3184 & 0.9469 \\
			\hline
		\end{tabular*}
	\end{center}
\end{table*}

The same simulations were performed with the adjusted algorithms (see Remark~\ref{rem:adjust}). One can see in Table~\ref{fintable_improved} that the accuracy is almost the same. By its definition, the adjustment always improves the approximation of the original series; however, the influence on the accuracy of signal approximation is ambiguous (the adjustment improves the accuracy at the 100-th iteration; results are various at the first iteration).

Thus, the numerical examples mostly confirm the statements itemized in Section~\ref{sec:comments}.

\section{Real-life example}
\label{sec:ex_real}
Let us consider the series `Fortified wine' (fortified wine sales, Australia, monthly, from January 1980 till December 1993) \citep{HyndmanTSDL}. This series has the following structure: a signal consisting of an exponential trend and a seasonality of a complex form and of noise.
We compare the Cadzow($\alpha$) algorithms for different $\alpha$ and demonstrate that a smaller $\alpha$ provides
a smaller approximation error.

In Section~\ref{sec:simul} we considered a simple example with a signal consisting of one sine wave.
However, the real-life time series has much more complex form. To confirm the approach that we can minimize
the approximation errors to diminish the signal estimation error,
let us construct a model of the `Fortified wine' series and use this model for simulation to
check the approach.

The series `Fortified wine' has been analyzed in several papers (see e.g. \cite{Golyandina.etal2015} for a bit longer time series).
A typical analysis of the time series by Basic SSA
\cite[Chapter 1]{Golyandina.etal2001}
with window length $L=84$ shows that the leading 11 eigentriples
correspond to the signal. The ESPRIT method~\cite{Roy.Kailath1989,Golyandina.Zhigljavsky2012} applied to the found signal subspace provides
estimates of exponential bases $\rho_m$ for the trend and for modulations of seasonal components in the series components, where the $k$-th term in the $m$th component is given in the form $C_m \rho_m^k$ or $C_m \rho_m^k \sin(2\pi\omega_m k +\phi_m)$, $k=1,\ldots,N$. The ESPRIT method also estimates the frequencies $\omega_m$; however, for seasonal components the possible frequencies are known and therefore we changed the frequency estimates to  nearest values in the form $j/12$.
The coefficients $C_m$ before the found series components and the phases $\phi_m$ of seasonal components were estimated by
the least-squares method. Noise is taken multiplicative, that is, its variance increases proportionally to the trend.
Thus, the model of the signal $\tsS = (s_{1}, \ldots, s_N)$, $N=168$, is estimated as
\begin{multline*}
s_{k} = 3997.74\, (0.9967)^k + \\
1174.75\, (0.9942)^k \sin(\frac{2\pi k}{12} - 2.249) + \\
425.75\, (1.0001)^k \sin(\frac{2\pi k}{4} + 2.333) + \\
211.55\, (1.004)^k \sin(\frac{2\pi k}{6} + 1.677) + \\
169.33\, (1.0007)^k \sin(\frac{2 \pi k}{2.4} + 1.533) + \\
361.07\, (0.9884)^k \sin(\frac{2 \pi k}{3} - 2.901).
\end{multline*}
The model of the whole series $\tsX=(x_1,\ldots,x_N)$ is
$x_i=s_i + 353.17\, (0.9967)^k \ve_i$,
where  $\ve_i$, $i=1,\ldots,N$,  is Gaussian white noise with mean equal to $0$ and variance equal to $1$.
We set $L=84$ and apply the Cadzow($\alpha$) algorithm with $\alpha=1, 0.8, 0.6, 0.4, 0.2, 0.1, 0.05$.
The following STOP1 criterion is taken in Algorithm~\ref{alg:obliqueCadzow}:
$\frac{\|\calT^{-1}(\bfY_k) - \calT^{-1}(\bfY_{k + 1})\|^2}{N} < 10^{-4}$.
The algorithm was applied to 1000 independent realizations of the model and also to the original `Fortified wine' series.
Table~\ref{tab:rltable} contains the RMSE of model signal estimation (the column $\tsS$),
the RMSE of model series approximation (the column $\tsX$) and the approximation accuracy for
`Fortified wines' series (the column $\tsX^*$).

\begin{table*}
	\caption{Comparison of the errors of signal estimation and series approximation using the Cadzow methods, for the `Fortified wine' series and the model realizations.}
\label{tab:rltable}
	
	\begin{tabular*}{\textwidth}{@{\extracolsep{\fill}}lcccc}
		\hline
		Method: & $\tsS$ & $\tsX$ & $\tsX^*$ \\
		\hline
		Cadzow, $\alpha = 1$ & 127.71 & 263.20 & 283.58 \\
		\hline
		Cadzow, $\alpha = 0.8$ & 127.18 & 262.98 & 283.25 \\
		\hline
		Cadzow, $\alpha = 0.6$ & 126.42 & 262.63 & 282.72 \\
		\hline
		Cadzow, $\alpha = 0.4$ & 125.39 & 262.06 & 281.77 \\
		\hline
		Cadzow, $\alpha = 0.2$ & 124.10 & 260.94 & 279.55 \\
		\hline
		Cadzow, $\alpha = 0.1$ & 125.09 & 260.52 & 276.70 \\
		\hline
		Cadzow, $\alpha = 0.05$ & 129.44 & 261.47 & 274.00 \\
		\hline
	\end{tabular*}
\end{table*}

Table~\ref{tab:rltable} shows that for $\alpha\in [0.2,1]$ a smaller approximation error yields a smaller reconstruction error. However, for the smaller values $\alpha$
the tendency is broken. Probably, small values of $\alpha$ do not provide a sufficient separability from noise
to converge toward the global minimum.

Figure \ref{fig:rl} depicts the approximation of the original `Fortified wine' series $\tsX^*$ obtained by the Cadzow($0.2$) algorithm.
The dotted line corresponds to the original series, while the solid line shows the finite-rank estimate of the signal of rank 11.
One can expect that the Cadzow($0.2$) algorithm provides one of the most accurate finite-rank estimates of the signal.

\begin{figure}[!hhh]
	\includegraphics[width = \columnwidth]{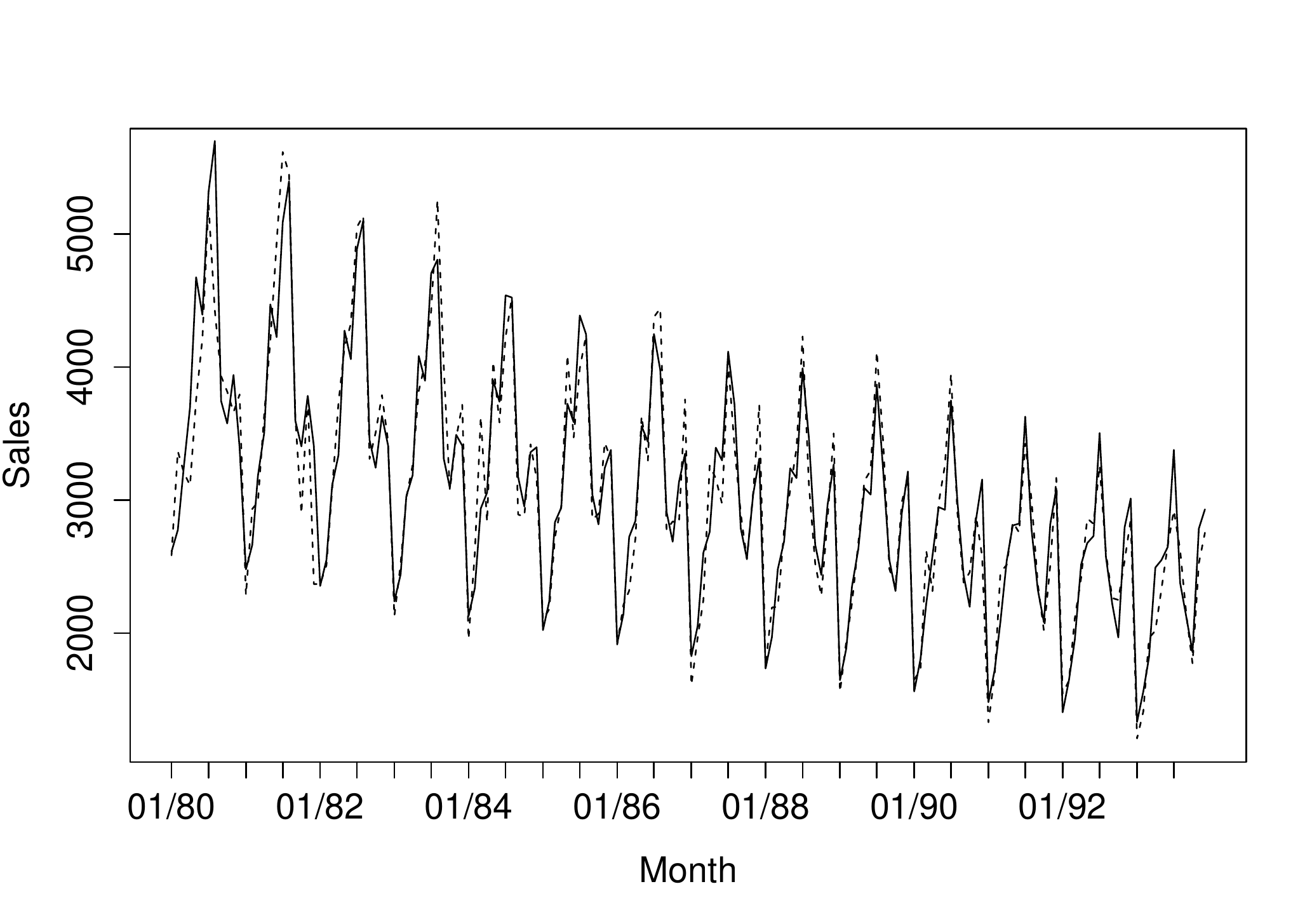}
	\caption{`Fortified wine' series: application of the Cadzow($0.2$) algorithm.}
	\label{fig:rl}
\end{figure}

\section{Conclusion}
\label{sec:concl}

Several known and new iterative algorithms for approximation of a noisy signal by a finite-rank series were considered in the present paper. The approximation was performed by a least-squares method and its result was considered as an estimate of the signal.

 We used equivalent statements of the problems for weighted matrix approximation and  weighted time-series approximation, where equal weights in the least-squares matrix problem correspond
 to unequal weights in the least-squares time series problem, and vice verse.

A wide range of the iterative algorithms was reviewed with the aim to obtain equal weights in the least-squares method applied to time series.
Equal weights were formally achieved in the algorithms using inner iterations, which converge to a local minimum only and also make the algorithms very time-consuming. It appears that the use of methods without inner iterations (Cadzow-type methods) leads to approximately equal weights only.

Convergence of outer iterations by subsequences was proved for the reviewed algorithms.

Comparison of the accuracy and convergence rate was performed by simulation on the example of a noisy sine-wave signal. The simulation results confirmed the theoretical results. It appears that time-series weights, which are closer to equal, provides in the limit more time-consuming and simultaneously more accurate methods.
Also, the simulations confirm that the convergence rate is in accordance with the separability rate. Therefore, for the Cadzow-type methods,
there is the correspondence between slow convergence, poor separability,
inaccurate approximation at the first one iteration and high accuracy in the limit; and vice verse.
In particular, for the Cadzow iterations, which produce Singular Spectrum Analysis for signal reconstruction at the first iteration, the window length equal to half of the series length gives poor accuracy in the limit and
one of the best reconstructions at the first iteration.

\appendix

\section{Separability of sine-wave signal from constant residual for the Cadzow($\alpha$) iterations}
\label{sec:app}

Let us consider modifications of SSA, which are produced by the first iteration of the Cadzow($\alpha$) iterative algorithms described in Section~\ref{sec:cadzow_alpha}. Recall that the Cadzow($1$) iterative algorithm produces the conventional Basic SSA method \cite{Golyandina.etal2001,Golyandina.Zhigljavsky2012}, while the first iteration of a general Cadzow($\alpha$) algorithm can be considered as a
particular case of Oblique SSA \cite{Golyandina2013} with the Euclidean inner product in the column space and a special inner product in the row space.

Separability of signals from residuals in SSA is deeply investigated in \cite{Golyandina.etal2001,Golyandina2010}. Separability of a signal means the ability of the method to extract the signal. In fact, separability is related to the accuracy of signal estimation obtained at the first iteration
of the considered iterative algorithms. Notions of exact, approximate and asymptotic (as the series length tends to infinity) separability together with examples of the asymptotic separability rates are introduced in \cite{Golyandina.etal2001} and can be generalized for the oblique case. Following by \cite{Golyandina.etal2001}, we will measure the separability by means of the cosines between $L$- and $K$-lagged vectors
of the signal and the residual.

Let $\bfC \in \sfR^{K \times K}$ be a symmetric positive semidefinite matrix, $\tsX_1$ and $\tsX_2$ be two different time series of length $N$, $\bfX^1$, $\bfX^2$ be their trajectory matrices. Define the so-called \emph{correlation coefficient between the $i$-th and $j$-th columns} as:
\begin{equation}\label{col_corr}
\rho^c_{i,j} = \frac{(X^1_i, X^2_j)}{\|X^1_i\| \|X^2_j\|},
\end{equation}
where $X^k_i$ is the $i$-th column of the matrix $\bfX^k$, $k = 1, 2$, $(\cdot, \cdot)$ is the Euclidean inner product, $\|\cdot\|$ is the Euclidean norm. Define the \emph{correlation coefficient between the $i$-th and $j$-th rows} as:
\begin{equation}\label{row_corr}
\rho^r_{i,j} = \frac{(X^{1,i}, X^{2,j})_\bfC}{\|X^{1,i}\|_\bfC \|X^{2,j}\|_\bfC},
\end{equation}
where $X^{k,i}$ is the $i$-th row of matrix $\bfX^k$, $k = 1, 2$, and $(\cdot, \cdot)_\bfC$ is the oblique inner product in $\sfR^K$ generated by a matrix $\bfC$ as follows: $(X, Y)_\bfC = X \bfC Y^\sfT$ (here $X$ and $Y$ are row vectors), $\|\cdot \|_\bfC$ is the norm with respect to this inner product. We say that the series $\tsX_1$ and $\tsX_2$ are \emph{weakly $\varepsilon$-separable} if
\begin{equation}\label{weak_sep_eq}
\rho = \max\Big(\max_{1 \le i,j \le K}|\rho^c_{i,j}|, \max_{1 \le i,j \le L}|\rho^r_{i,j}|\Big) < \varepsilon.
\end{equation}
We are interested in the order of $\varepsilon$ as $N\ra \infty$ for different matrices $\bfC$, where the series $\tsX_k$, $k=1,2$, consist of the first $N$ terms of infinite series $\tsX_k^\infty$.

Here we apply the theory to an example with a sine-wave signal and a constant residual. By analogy with SSA, we can expect that the asymptotic separability rate will be the same if the residual is Gaussian white noise.
Thus, let
$\tsX_1^\infty = (\cos(2 \pi \omega k), k = 1, 2, \ldots)$ and  $\tsX_2^\infty = (c, c, \ldots)$.
Consider $N \to \infty$ and $L(N),\,K(N) \to \infty$ such that $N = L + K - 1$.
When $\bfC$ is the identity matrix, the answer is known: $\varepsilon$ has order $1/\min(L,K)$, i.e. the rate of separability has order $1/N$ for $L$ proportional to $N$.
This result can be found in \cite[Section 6.1]{Golyandina.etal2001}.

\smallskip
Let us consider the separability rate for the Cadzow($\alpha$) iterations introduced in Section~\ref{sec:cadzow_alpha}.

\begin{remark}
	In what follows we will use the following denotation:
a function $f \in O(g(n))$ as $n\ra \infty$ if there exist $k>0$ and $n_0>0$ such that for any $n > n_0$ the inequality $|f(n)| \le k |g(n)|$ holds;
a function $f \in \Omega(g(n))$ as $n\ra \infty$ if there exist $k>0$ and $n_0>0$ such that for any $n > n_0$ the inequality $|f(n)| \ge k |g(n)|$ holds.
\end{remark}

\begin{proposition}
	\label{prop:separ1}
	Let $\tsX_1^\infty = (\cos(2 \pi \omega k), k = 1, 2, \ldots)$, where $0<\omega <0.5$, be a sine wave, $\tsX_2^\infty = (c, c, \ldots)$ be a constant series,  $L(N),K(N)\ra \infty$, where $N=L+K-1$, $h = h_N = \lfloor N/L \rfloor$. Let also $0<\alpha=\alpha(N)\le 1$, and $\bfC=\bfC(\alpha)$ be defined in \eqref{zhigweights}, i.e.  $\bfC$ is a diagonal matrix with diagonal elements:
	\begin{equation*}
	c_k = \begin{cases}
	1, & \text{if} \quad k = jL+1 \quad \text{for some} \ j = 0, \ldots, h-1,\\
	\alpha, & \text{otherwise},
	\end{cases}
	\end{equation*}
Then
	\begin{enumerate}
		\item $\rho$ given by \eqref{weak_sep_eq} has the following order: $\rho~=~O\left(\max\left(\frac{1}{L}, \frac{(1-\alpha)C_{L,K}+\alpha}{(1-\alpha)D_{L,K}+\alpha K}\right)\right)$, where
		\begin{equation*}
		C_{L,K} = C_{L(N),K(N)} = \max_{\substack{1 \le j \le L}}  \sum_{\substack{1 \le k \le K: \\ c_k = 1}}\cos(2 \pi \omega (j + k - 1)),		\end{equation*} and
		\begin{equation*}
		D_{L,K} = D_{L(N),K(N)} = \min_{\substack{1 \le j \le L}} \sum_{\substack{1 \le k \le K: \\ c_k = 1}}\cos^2(2 \pi \omega (j + k - 1)).
		\end{equation*}
		\item If $h_N$ is bounded by a constant, then $\rho=O\left(\max\left(\frac{1}{L}, \frac{1}{\alpha K}\right)\right)$.
		\item If there exists small $\delta$, $0 < \delta < 1/2$, such that $2\,L(N)\,\omega\in\sfR \setminus \left(\bigcup_{k \in \sfZ} [k - \delta, k + \delta] \right)$ for every $N$, where $\sfZ$ is the set of integers, then $\rho~=~O\left(\max\left(\frac{1}{L}, \frac{1}{(1-\alpha)N/L+\alpha K}\right)\right)$.
	\end{enumerate}
	
\end{proposition}

\begin{proof}
1. To prove the theorem, we should evaluate the order of the expressions:
	\begin{equation}\label{th5_sep1}
	\rho^c_{i,j} = \frac{\sum_{k=j}^{j + L - 1} \cos(2 \pi \omega k)}{\sqrt{L \left(\sum_{k=j}^{j + L - 1} \cos^2(2 \pi \omega k)\right)}},
	\end{equation}
	\begin{equation}\label{th5_sep2}
	\rho^r_{i,j} = \frac{\sum_{k=1}^K c_k\cos(2 \pi \omega (j + k - 1))}{\sqrt{\left(\sum_{k=1}^K c_k\right) \left(\sum_{k=1}^K c_k\cos^2(2 \pi \omega (j + k - 1))\right)}}.
	\end{equation}
	The following trigonometric equalities hold:
	\begin{equation}
	\label{sumcos}
	\sum_{k=1}^n \cos(ak + b) = \csc(a/2) \sin(an / 2) \cos \left(\frac{an + a + 2b}{2} \right),
	\end{equation}
	\begin{multline}
	\label{sumsqcos}
	\sum_{k=1}^n \cos^2(ak + b) = \frac{1}{4}(2n + \csc(a) \sin(2an + a + 2b) -\\ - \csc(a)\sin(a + 2b)),
	\end{multline}
	for any real $a, b$ and positive integer $n$.
	Therefore, since the series $\tsX_1$ is not constant, the numerator in \eqref{th5_sep1} has order $O(1)$, while the denominator has order $\Omega(L)$. Thus, we obtain the order $1/L$.
	
	To evaluate the order of \eqref{th5_sep2}, consider the sum over $k$ such that $c_k=1$ separately:
	\begin{multline*}
	\sum_{k=1}^K c_k\cos(2 \pi \omega (j + k - 1)) = \\ (1-\alpha) \sum_{\substack{1 \le k \le K: \\ c_k = 1}}\cos(2 \pi \omega (j + k - 1)) +\\ +\sum_{1 \le k \le K}\alpha \cos(2 \pi \omega (j + k - 1)) = (1-\alpha)O(C_{L,K}) + \alpha\, O(1),
	\end{multline*}
    \begin{equation*}
	\sum_{k=1}^K c_k = (1-\alpha) h + \alpha K,
	\end{equation*}
and
    \begin{multline*}
	\sum_{k=1}^K c_k\cos^2(2 \pi \omega (j + k - 1)) = \\ (1-\alpha)\sum_{\substack{1 \le k \le K: \\ c_k = 1}}\cos^2(2 \pi \omega (j + k - 1)) +\\ +\sum_{1 \le k \le K }\alpha \cos^2(2 \pi \omega (j + k - 1)) = (1-\alpha) \Omega(D_{L,K}) + \alpha\, \Omega(K).
    \end{multline*}

    2. $C_{L,K}$ is exactly the maximum of sums, each of $h$ cosine values, therefore, the absolute value of $C_{L,K}$ is not larger than $h$. Therefore, if $h$ is bounded by a constant, then $|C_{L,K}|$ is bounded by the same constant, so, $C_{L,K} = O(1)$.

    3. The condition $2\, L(N)\, \omega \in  \sfR \setminus \left(\bigcup_{k \in \sfZ} [k - \delta, k + \delta] \right)$ guarantees that  $|\csc(\pi L(N) \omega)|$ in \eqref{sumcos} for $C_{L,K}$ and $|\csc(2 \pi L(N) \omega)|$ in \eqref{sumsqcos} for $D_{L,K}$ are bounded by a constant; therefore, we obtain an upper bound for $C_{L,K}$ and a lower bound for $D_{L,K}$. Thus, $C_{L, K}$ has order $O(1)$, while $D_{L, K}$ has order $\Omega(N/L)$.
\end{proof}

\begin{remark}
	Let us suppose that we have chosen $L(N)$ such that $\rho$ has order $\max\left(\frac{1}{L}, \frac{1}{(1-\alpha)N/L+\alpha K}\right)$. Then the optimal choice for $L$ is $L \approx \frac{\alpha(N + 1) + \sqrt{\alpha^2(N+1)^2 + 4N(1  - \alpha^2)}}{2(1 + \alpha)}$.
Hence, the rate of separability has the same order $O(1/N)$  for $\alpha(N) \to c$, where $0<c\le 1$ is some constant (however, a smaller $c$ corresponds to a smaller multiplier before $1/N$). In the case of converging to zero $\alpha(N) = O(N^{-\beta})$, the rate of separability becomes equal to $O(N^{\beta - 1})$ for $0 \le \beta \le 0.5$ and to $O(1/\sqrt{N})$ for $\beta > 0.5$.
\end{remark}

\bibliographystyle{plain}


\end{document}